\documentclass[11pt]{amsart}
\UseRawInputEncoding

\usepackage{amsmath, amssymb, amsthm, amscd}
\usepackage{url}
\usepackage{graphicx}
\usepackage[hidelinks,pagebackref,pdftex]{hyperref}
\usepackage{color}
\usepackage{import}
\usepackage[letterpaper, margin=1 in]{geometry}
\usepackage{caption}
\usepackage{bbold}
\usepackage{subcaption}

\usepackage[ruled,vlined, linesnumbered]{algorithm2e} 
\SetKwInput{KwRequire}{Require}

\usepackage{amsmath}
\usepackage{old-arrows}
\usepackage{comment}
\usepackage{multirow}
\usepackage{tikz-cd}
\usepackage{enumitem}
\usepackage{stmaryrd}
\usepackage{epigraph}
\usepackage{braket}
\newcommand{\Z}{\mathbb{Z}}

\newcommand{\R}{\mathbb{R}}
\newcommand{\C}{\mathbb{C}}

\newcommand{\eg}{{\em e.g.}}
\newcommand{\ie}{{\em i.e.}}

\DeclareMathOperator{\SU}{SU}

\DeclareMathOperator{\tr}{tr}

\DeclareMathOperator{\diag}{diag}
\DeclareMathOperator{\su}{\mathfrak{su}}
\newcommand{\calG}{\mathcal{G}}

\DeclareMathOperator{\poly}{poly}

\usepackage{marginnote}
\long\def\@savemarbox#1#2{\global\setbox#1\vtop{\hsize\marginparwidth 
  \@parboxrestore\tiny\raggedright #2}}
\renewcommand*{\backref}[1]{}
\renewcommand*{\backrefalt}[4]{
  \ifcase #1
  [No citations.]
  \or [#2]
  \else [#2]
  \fi }

\AtBeginDocument{%
   \def\MR#1{}
}

\numberwithin{equation}{section}
\theoremstyle{plain}
\newtheorem{theorem}[equation]{Theorem}

\newtheorem{lemma}[equation]{Lemma}
\newtheorem{corollary}[equation]{Corollary}
\newtheorem{proposition}[equation]{Proposition}

\newtheorem*{namedtheorem}{\theoremname}
\newcommand{\theoremname}{testing}

\theoremstyle{definition}
\newtheorem{definition}[equation]{Definition}

\newtheorem{remark}[equation]{Remark}

\title[Removing Online Exponential Net Search from Solovay-Kitaev]{Removing Online Exponential Net Search from Solovay-Kitaev}

\author{Henrique Ennes}\thanks{Universit\'e C\^ote d'Azur, Inria, CNRS, I3S, Sophia Antipolis, France, \texttt{henrique.lovisi-ennes@inria.fr}.}
\author{Cl\'ement Maria}\thanks{Inria Center d'Universit\'e C\^ote d'Azur, Sophia Antipolis, France, \texttt{clement.maria@inria.fr}.}
\begin{document}

\pagenumbering{gobble}
\begin{abstract}
The Solovay--Kitaev algorithm describes how to approximate, to arbitrary precision, a matrix in the special unitary group $\SU(d)$ using any fixed universal gate set.
Although the algorithm scales as $O(\poly(\log(1/\varepsilon)))$, where $\varepsilon$ is the maximum targeted approximation error, its running time depends exponentially on the qudit dimension $d$.
This bad dependence can be traced to its explicit use of an $\varepsilon_0$-net of size $2^{\Omega(d^2)}$, which is queried $O(\poly(\log(1/\varepsilon)))$ times throughout the execution.
For this reason, the standard Solovay--Kitaev theorem is usually stated for fixed $d$, with the base net and its lookup cost absorbed into the constants.
We study the algorithmic problem in the variable-dimension regime and show how to avoid searching an exponentially large precomputed net for each target unitary.
In particular, we introduce the notion of a good exponential basis and show that such a basis can replace the usual depth-zero net-search routine.
This yields a modification of the algorithm in which the use of an explicit net is fully moved to a preprocessing step.
For instruction sets that already contain, or allow the efficient construction of, a good exponential basis, the resulting online synthesis algorithm is polynomial in $d$ and polylogarithmic in $1/\varepsilon$.
For arbitrary universal instruction sets, the exponential dependence on $d^2$ is not removed, but is isolated into a one-time additive preprocessing cost.
Our technique uses differential-geometric methods to devise an integerized version of trotterization that replaces the depth-zero net query by a constructive local synthesis routine.
The same framework also suggests possible extensions based on other discretized numerical integration schemes.
\end{abstract}

\maketitle
\pagenumbering{arabic}
\section{Introduction}
\epigraph{\emph{The presence of the exponential $\exp(O(M^2))$ in the algorithm
complexity bound is rather disturbing [...]. As far as the asymptotic behavior at $\varepsilon\to 0$ is
concerned, it seems possible to make the computation polynomial in $M$,
that is, the exponential may become an additive term rather than a factor.
(To this end, one may try to use bases in the tangent space instead of nets
— the reader is welcome to explore this idea.) However, it is a challenge
to eliminate the exponential altogether. This may be only possible if one
changes the assumptions of the theorem, \eg, by saying that products of
$\poly(M)$ elements from $\mathcal{N}_{\varepsilon_0}$ constitute an $\varepsilon_0$-net (rather than $\mathcal{N}_{\varepsilon_0}$ being an $\varepsilon_0$-net itself). Such a $\mathcal{N}_{\varepsilon_0}$ can consist of only $\poly(M )$ elements, so it is
reasonable to ask whether there is an approximation algorithm with running
time $\poly(M \log(1/\varepsilon))$. This appears to be a difficult question in global unitary geometry.}}{Alexei Kitaev, Alexander Shen, and Michael Vyalyi, \emph{Classical and Quantum Computation}}

\emph{Quantum computing} was born from the desire to leverage the laws of quantum mechanics to perform tasks that are usually believed to lie outside the realm of efficient computation by digital machines.
These devices are modeled by sequences of unitary transformations, called \emph{gates}, acting on finite-dimensional complex Hilbert spaces $\C^d$.
The normalized vectors of these spaces---more precisely, the corresponding rays---describe pure states of physical quantum systems, called \emph{qubits} when $d=2$, or, more generally, \emph{qudits}.
Any realization of quantum computing comes with only a finite collection of implementable gates, $\{G_i\}_{i=1}^m$, called an \emph{instruction set}, which can be physically assembled into circuits and is \emph{universal}, meaning that these gates can be combined to approximate any unitary in $\SU(d)$ to arbitrary accuracy.
We note that universality is only an existence property, and we are naturally led to consider the corresponding compilation problem: given an input target unitary $U$, find a product of gates in $\{G_i\}_{i=1}^m$ that approximates $U$ within a prescribed error $\varepsilon$.
We call this the \emph{gate synthesis problem}.
Here, we are interested not only in the length of the resulting product, which we naturally desire to be small because it measures the resources required for the circuit implementation of the target unitary, but also in the classical time and space required to find it, as functions of both $d$ and $\varepsilon$.

It is not hard to imagine an algorithm that outputs circuits with $O(1/\varepsilon)$ gates and uses no ancillary qudits.
Nevertheless, even at the dawn of quantum computing, Deutsch, Barenco, and Ekert~\cite{deutsch1995universality} already conjectured that more efficient solutions to the gate synthesis problem, with gate count and total running time polylogarithmic in $1/\varepsilon$, could be possible.
We stress that the desire for this sort of asymptotic improvement can be justified by sensible requirements for \emph{quantum advantage}.
Dawson and Nielsen~\cite{dawson2005solovay} point out, for example, that the expected quadratic speedup of Grover's search algorithm~\cite{grover1996fast, grover1997quantum} is only meaningful if one assumes sublinear dependence of the gate count on the inverse of the computational accuracy.
Perhaps an even more daunting problem is that computations in the complexity class $\texttt{PostBQP}$ assume the ability to approximate circuit probabilities that decrease exponentially with the input size~\cite{aaronson2005quantum, aaronson2005complexity, alagic2017quantum, kuperberg2015hard}.
Consequently, the class can only be defined consistently and independently of the gate set if transformations between different gate sets incur at most a polylogarithmic overhead in the inverse of the accuracy.

In 1995, Solovay outlined, on an email list, an algorithm that approximates any unitary acting on qubits in classical $O(\log^k(1/\varepsilon))$ time using $O(\log^k(1/\varepsilon))$ gates from a universal instruction set, for some $k>1$.
This algorithm was later formalized and extended by Kitaev to qudits~\cite{kitaev1997quantum}.
The Solovay--Kitaev (SK) algorithm, as it is now known, uses the geometric structure of the Lie group of unitary gates, $\SU(d)$, to provide an ancilla-free solution to gate synthesis, where $d$ is the qudit dimension, that is, the dimension of the Hilbert space modeling its states.
Explicitly, it recursively improves a coarse initial approximation of the target consisting of $\ell_0$ gates and having error at most $\varepsilon_0$, producing circuits whose approximation errors decrease superlinearly with the recursion depth while their lengths grow geometrically.
Although this algorithm is almost as efficient as possible as a function of $\varepsilon$ alone---with modifications nearly reaching the information-theoretic bound of $O(\log(1/\varepsilon))$~\cite{kuperberg2023breaking}---it still requires a sufficiently good initial coarse approximation, within distance $\varepsilon_0$, for the recursion to kick in.

This coarse approximation is obtained from an $\varepsilon_0$-\emph{net}, $\mathcal{N}_{\varepsilon_0}$, consisting of circuits of depth at most $\ell_0$ such that every unitary in $\SU(d)$ lies within distance $\varepsilon_0$ of some net node~\cite{alon1987partitioning, mustafa2017epsilon}.
In the standard implementation, such a net is constructed by explicitly enumerating circuits up to depth $\ell_0$.
The net size required for SK to converge is independent of the final desired accuracy $\varepsilon$, but depends strongly on the qudit dimension $d$: for fixed $\varepsilon_0$, the covering number of $\SU(d)$ scales as $2^{\Omega(d^2)}$~\cite{dawson2005solovay}.
This exponential dependence of the net size on $d^2$ implies not only unreasonable memory requirements for high-dimensional qudits, but also a huge prefactor in the SK running time, since searches over the net are necessary at every depth-zero recursive call for each input target $U$.
Conversely, improving the dependence on the dimension is more than a theoretical curiosity: it is essential for the practical realization of several models of quantum computing, such as photonic systems~\cite{bente2025potential, wang2020qudits, wang2023photonic}, trapped ions~\cite{bruzewicz2019trapped, haffner2008quantum, hrmo2023native}, and topological quantum computing~\cite{das2006topological, freedman2003topological, stern2013topological, wang2010topological}.
Although metric-space similarity-search methods~\cite{zezula2006similarity}---\eg, geometric near-neighbor access trees~\cite{brin1995near}, kd-trees~\cite{bentley1975multidimensional, de2008computational}, ball trees~\cite{omohundro1989five}, and locality-sensitive hashing~\cite{indyk1998approximate}---can reduce typical query costs, they do not remove the exponential dependence on the size of the explicit net in the worst-case setting relevant here.
Ultimately, the use of nets in SK exposes the algorithm to the well-known \emph{curse of dimensionality}~\cite{bellman1966dynamic}.

This paper was born from the desire to investigate a version of the SK algorithm in which $\mathcal{N}_{\varepsilon_0}$ is \emph{only implicitly used}.
Thanks to the manifold structure of $\SU(d)$, such a net-free version of Solovay--Kitaev might be possible, but, as pointed out by Kitaev, Shen, and Vyalyi~\cite{kitaev2002classical}, an algorithmic construction seems far from trivial.
Still, these authors suggest that, by exploiting the geometric structure of $\SU(d)$, it might already be feasible to construct a version of the SK algorithm in which the original running time
\begin{equation*}
O\left((d^3+T_0)p\log^{k_t}(1/\varepsilon)\right)
\end{equation*}
is replaced by something of the form
\begin{equation}\label{eq: p and t0}
O\left(\poly(d)T_0+\poly(d)p\log^{k_t}(1/\varepsilon)\right),
\end{equation}
where $T_0$ denotes the cost of the depth-zero net-search routine and $p$ denotes the number of target unitaries that we wish to approximate within an error of $\varepsilon$ each.
The additive dependence between $p$ and $T_0$ in equation~\eqref{eq: p and t0} makes this alternative version algorithm better suited for \emph{online applications} where a stream of unitaries is given as input, which is precisely the setting most often encountered in practical quantum computing, where circuits must be transpiled from one gate instruction set into another.

To the best of our knowledge, this suggestion has not previously been developed into an explicit algorithm.
In this paper, we give such a construction.
Our method treats the dependence on the coarse net $\mathcal{N}_{\varepsilon_0}$ as a \emph{preprocessing} cost: after an initial phase in which the net is explicitly constructed and queried, the recursive synthesis phase proceeds without further net searches.
Informally, the usual multiplicative dependence on the depth-zero net-search cost $T_0$ in the SK recursion is replaced by an additive preprocessing term.
From this point on, subsequent targets in $\SU(d)$ can be synthesized in time polynomial in $d$ and polylogarithmic in $1/\varepsilon$, with no additional access to $\mathcal{N}_{\varepsilon_0}$.

In fact, for some suitable choices of instruction sets, the preprocessing step can be fully avoided.
Explicitly, given a target unitary $U$, our technique uses a bi-invariant Riemannian metric on $\SU(d)$ to compute a geodesic from the identity $I$ to $U$ and approximate it using the gates of the instruction set.
Our choice of a bi-invariant metric is motivated by the fact that its geodesics can be approximated to arbitrary accuracy using \emph{product formulas}.
These approximation techniques allow us to derive a version of the SK algorithm, Algorithm~\ref{alg: modified sk}, that is polynomial in $d$, provided that we use a special instruction set
$\{\calG_j\}_{j=1}^M$, which we call a \emph{good exponential basis}.
However, because SK is most useful when applied to an arbitrary instruction set, we proceed to describe an algorithm that transforms a general instruction set into a good exponential basis.
Constructing $\{\calG_j\}_{j=1}^M$ from the original gates $\{G_i\}_{i=1}^m$ is possible with the standard SK algorithm, and this is exactly the preprocessing step in which a net is still used.
The overall complexity is summarized in Corollary~\ref{col: sk}.
We do not know whether this was the construction that the authors of~\cite{kitaev2002classical} had in mind, but it has the advantage of modifying SK only in the subroutine in which $\mathcal{N}_{\varepsilon_0}$ is searched.
In particular, our framework is compatible with modern improvements to the SK algorithm, such as the inverse-free version of~\cite{bouland2021efficient} and the improved asymptotics of~\cite{kuperberg2023breaking}, as well as with more accurate numerical integration schemes for geodesics in $\SU(d)$.

This paper is divided as follows.
We start in Section~\ref{sec: geometry quantum} by describing a geometric-flavored formulation of quantum computing.
In this, no further technical knowledge of quantum mechanics or quantum computing beyond what was laid in this introduction will be required from the reader.
Most of the work in Section~\ref{sec: geometry quantum} will be in establishing the BCH formula, which will be a crucial ingredient to establish the SK Theorem in Section~\ref{sec: sk section big}.
In Section~\ref{sec: lifting}, our modification of the algorithm is presented, and its correctness and running complexity are demonstrated.
We finish with a discussion where our methods are compared with the usual SK.
\vspace{0.3 cm}

\paragraph{\textbf{Conventions:}}
As is customary in the literature on the Solovay--Kitaev algorithm, our complexity bounds are stated in the real-RAM model: arithmetic operations and comparisons on real numbers are assumed to take unit time.
We therefore do not track the bit precision required to implement scalar arithmetic, matrix logarithms, or spectral decompositions.
For a discussion of the role of precision in the use of SK in complexity theory, see~\cite{aaronson2014postbqp} and the comments therein.
We expect that the estimates in this paper can be made stable under finite precision, given explicit bounds on the relevant condition numbers, but we do not pursue such an analysis here.

We also use the standard straight-line-program representation for synthesized circuits.
Thus, the output is represented by a sequence of previously constructed gates and subroutines rather than by an explicitly expanded word over the original instruction set.
Our space bounds refer to the working memory required to construct this representation; see~\cite{dawson2005solovay} for a related discussion.

Finally, we assume that every instruction set is closed under inverses, that is, if $G_i$ is an implementable gate, then so is $G_i^\dagger$.
A version of the SK algorithm that does not assume direct access to inverses was given in~\cite{bouland2021efficient}, but it has higher running-time costs.

\vspace{0.3 cm}
\paragraph{\textbf{Acknowledgments:}}
This work has been partially supported by the ANR project ANR-20-CE48-0007 (AlgoKnot) and the project ANR-15-IDEX-0001 (UCA JEDI). It has also been supported by the French government, through the France 2030 investment plan managed by the Agence Nationale de la Recherche, as part of the ``UCA DS4H'' project, reference ANR-17-EURE-0004.
We are much in debt to Chih-Kang Huang for all the detailed answers to our many questions on quantum control theory; this paper would probably not exist if it were not for his immense help. 

\section{Geometry of quantum computing}\label{sec: geometry quantum}
In the usual model for pure state quantum mechanics assumed in this paper, the evolution of some normalized state $\ket{\psi}\in\C^d$ is described by a unitary operator $U$ in the \emph{special unitary group} $\SU(d)$, meaning that $U^\dagger U =I$, for $I$ the complex $d\times d$ identity matrix and with $\det U = +1$.
Unitarity of $U$ implies that its eigenvalues are of form $\mu_j = e^{i\theta_j}$ for some $\theta_j\in \R$, whereas the constraint on the determinant gives
\begin{equation}\label{eq: mod eigenvalues}
    \sum_j \theta_j =0 \; \mod 2\pi.
\end{equation}
As a subset of the $d\times d$ complex matrices, $\SU(d)$ naturally inherits the operator norm, which we denote by $\|\cdot\|$.
In particular, because special unitary matrices are normal, $\|U\|$ equals the largest absolute value of an eigenvalue of $U\in\SU(d)$, which is $1$.

The product and inversion rules of $\SU(d)$ are smooth when giving it a (real) manifold structure of dimension $d^2-1$. 
In particular, $\SU(d)$ is a 
\emph{compact Lie group}.
We will not review here the many important properties of compact Lie groups (the interested reader might refer to~\cite{brocker2003representations}), but we will often use their tangent space at the identity, $T_I\SU(d)$, which we denote by $\su(d)$.
Naturally, this forms a real vector of dimension $d^2-1$, and the constraints of $\SU(d)$ can be used to identify $\su(d)$ with the vector space of Hermitian ($ X^\dagger = X$) and traceless ($\tr X = 0$) matrices, which is closed under \emph{commutators} $[X,Y]=i(XY-YX)$, where multiplication is taken with respect to the usual matrix product.
When endowed with the commutator $[\cdot,\cdot]$, $\su(d)$ forms a \emph{Lie algebra}~\cite{erdmann2006introduction, hall2013lie} and we shall call it \emph{the} Lie algebra of $\SU(d)$.
For later reference, we will note that the induced Lie algebra structure implies that $[\cdot,\cdot]$ is bilinear (with respect to real linear combinations) and anti-commutative.
We can use matrix exponentiation to define a smooth map from the Lie algebra to the group. 
We follow the standard Taylor series to write
\begin{equation}\label{eq: matrix exponential}
    \exp(iX)=\sum_{n=0}^\infty \frac{(iX)^n}{n!}
\end{equation}
where the factor of $i$ is introduced to guarantee that the exponential converges to an element of $\SU(d)$.

We will find it useful to also define a map from the Lie group to the algebra as a right-inverse of the exponential function.
Its definition and existence are given through the lemma below.

\begin{lemma}\label{lm: new bound}
    Suppose $U\in\SU(d)$.
    Then there exists an $H\in\su(d)$ such that
    \begin{equation*}
        U = \exp(iH)\quad \text{and}\quad \|H\|\leq 2\pi.
    \end{equation*}
    In particular, we define the \emph{logarithmic} of $U$ as $\log U = H$.
\end{lemma}

\begin{proof}
    By the spectral theorem, there exists some matrix $V\in U(d)$ such that
    \begin{equation*}
        U = V \diag(e^{i\widetilde{\theta_1}},\dots,e^{i\widetilde{\theta_d}}) V^\dagger,
    \end{equation*}
    with $\widetilde{\theta_j}$ as in equation~\eqref{eq: mod eigenvalues}.
    In particular one may assume that $\widetilde{\theta_j}\in(-\pi,\pi]$.
    The matrix
    \begin{equation*}
        \widetilde{H} = V \diag(\widetilde{\theta_1},\dots,\widetilde{\theta_d})V^\dagger
    \end{equation*}
    is called the \emph{principal branch logarithm} of $U$.
    The spectral theorem implies that $\exp(i\widetilde{H})=U$, where $\exp$ is taken as in equation~\eqref{eq: matrix exponential}.
    We note that, while $\widetilde{H}$ is hermitian, the choice of $|\widetilde{\theta_j}|\leq \pi$ makes that it is not traceless, but only that
    \begin{equation*}
        \tr \widetilde{H} = \sum_j \widetilde{\theta_j} = 2k \pi,
    \end{equation*}
    for some $k\in\Z$. 

    We will choose some $\theta_j$ such that
    \begin{equation*}
         {H} = V \diag({\theta_1},\dots,{\theta_d})V^\dagger
    \end{equation*}
    is traceless, but $\exp(iH)=U$.
    If $k=0$, let $\theta_j=\widetilde{\theta_j}$ and we are done.
    If $|k|>0$, there exists at least $k$ values of $\widetilde{\theta_j}$ that are positive.
    Randomly choosing $k$ of them to be such that
    \begin{equation*}
        \theta_j = \widetilde{\theta_j} - 2\pi,
    \end{equation*}
    whereas, for the rest, we simply let $\theta_j=\widetilde{\theta_j}$,
    we note that
    \begin{equation*}
        \sum_j \theta_j = 0\quad\text{and}\quad|\theta_j|\leq 2\pi.
    \end{equation*}
    But as $\diag(e^{i\widetilde{\theta_1}},\dots,e^{i\widetilde{\theta_d}})=\diag(e^{i{\theta_1}},\dots,e^{i{\theta_d}})$, the result follows.
    Similarly, if $k<0$, we add $2\pi$ to $|k|$ of the negative $\widetilde{\theta_d}$.
\end{proof}
\begin{remark}
    It should be noted that, provided $\|H\|\leq \pi$, our definition of logarithmic agrees with the principle branch, refer to item~\ref{it: norm on log} of Lemma~\ref{lm: properties operator norm}.
\end{remark}

Many of the usual properties of the exponential function from scalar calculus are translated to equation~\eqref{eq: matrix exponential}. For example, $\exp(iH)^n = \exp(inH)$ for any $n\in \mathbb{Z}$, and $\frac{d}{dt}\exp(itH)=iH\exp(itH)$.
Nevertheless, in general, because matrix multiplication is non-commutative, $\exp(X)\exp(Y)$ does not equal $\exp(X+Y)$.
These quantities can be connected, on the other hand, by the Baker-Campbell-Hausdorff (BCH) formula~\cite{hall2013lie}, which we state in the lemma below.

Many of the usual properties of the exponential function from scalar calculus are translated to equation~\eqref{eq: matrix exponential}. For example, $\exp(iH)^n = \exp(inH)$ for any $n\in \mathbb{Z}$, and $\frac{d}{dt}\exp(itH)=iH\exp(itH)$.
Nevertheless, in general, because matrix multiplication is non-commutative, $\exp(X)\exp(Y)$ does not equal $\exp(X+Y)$.
These quantities can be connected, on the other hand, by the Baker-Campbell-Hausdorff (BCH) formula~\cite{hall2013lie}, which we state in the lemma below.

\begin{lemma}[BCH formula]\label{lm: bch}
Let $X_1,\dots,X_M \in \mathfrak{su}(d)$ be Hermitian operators satisfying
$$\sum_{j=1}^M\|X_j\| \le \delta,$$
for some $\delta \le \delta_0$, where $\pi > \delta_0 > 0$ is a sufficiently small universal constant.
Then the product admits the factorized expansion
\begin{equation}\label{eq: bch 1}
    \log\left[\prod_{j=1}^M \exp(i X_j)\right]
=
\sum_{j=1}^M X_j + \frac{1}{2} \sum_{i<j} [X_i, X_j] + R ,
\end{equation}\label{eq: bch 2}
where the remainder $R \in \mathfrak{su}(d)$ satisfies the norm bound
$$
\|R\| \leq C \delta^3,
$$
for some universal constant $C \geq 1$ independent of $d$ and $M$.
\end{lemma}

\begin{remark}
Letting $\delta_0 \leq \log 2$ is enough to guarantee convergence~\cite[Proposition 2.2]{biagi2020baker}.
We could also actually estimate this constant if we used the full BCH expansion, but we will not attempt to do so here. We note that we only assume $C\geq 1$ for later convenience.
\end{remark}

\begin{remark}\label{rm: bch}
     After increasing $C$ if necessary, in the special case $ X_1=H+E$ and $X_2= H$, the remainder also satisfies the refined estimate
    $$
    \|R\|
    \leq C\left(\|H\|+\|E\|\right)^2\|E\|.
    $$
    We will need this version of BCH in the proof of Theorem \ref{th: GM}.
\end{remark}

Besides BCH, we will also amply use the following bounds on the operator norm.

\begin{lemma}\label{lm: properties operator norm}
Suppose that $U,V\in \SU(d)$ are unitary, $A,B\in\su(d)$ Hermitian. Then
\begin{enumerate}[label=(\alph*)]
    \item\label{it: unitary invariance}
    $\|UAV\|=\|A\|;$
    \item\label{it: exponent approximation}
    $\|e^{iA}-e^{iB}\|\leq \|A-B\|;$
    \item\label{it: norm on log}
      Suppose $\|A\|\leq \pi$, then 
    $$\|A\|\leq \frac{\pi}{2}\|e^{iA}-I\|.$$
\end{enumerate}

\end{lemma}

\begin{proof}
   Item~\ref{it: unitary invariance} holds for any $d\times d$ matrix $A$, but we will only need the Hermitian case.
    The result follows by the invariance of the inner-product under the actions of the unitary group.
    Explicitly,
    \begin{equation*}
        \|UA\|^2=\sup_{\ket{\psi}}\bra{\psi}A^\dagger U^\dagger U A\ket{\psi} = \sup_{\ket{\psi}}\bra{\psi}A^\dagger A\ket{\psi}=\|A\|^2.
    \end{equation*}
    Similarly, we note that multiplication on the left by any unitary matrix is transitive on the unit sphere of $\C^d$, implying that
    \begin{equation*}
        \|AV\|^2=\sup_{\ket{\psi}} \bra{\psi}V^\dagger A^\dagger A V\ket{\psi}=\sup_{\ket{\psi'}}\bra{\psi'}A^\dagger A \ket{\psi'}=\|A\|^2.
    \end{equation*}
    
    For item~\ref{it: exponent approximation} we define
    \begin{equation*}
        f(t)=\exp[i (B+ t(A-B))]=\exp[iX(t)]. 
    \end{equation*}
    Duhamel's formula~\cite{higham2008functions} gives that
    \begin{equation*}
    \begin{split}
         \frac{d}{dt} f(t)=i \int_{0}^{1}
        e^{i(1-s)X(t)} \, (A-B) \, e^{isX(t)} \, ds 
    \end{split}
    \end{equation*}
    so
      \begin{equation*}
    \begin{split}
        \|e^{iA}-e^{iB} \|&= \left\|\int_0^1\frac{d}{dt} f(t)dt\right\|\\
        &\leq \int_{0}^{1}\int_{0}^{1}
        \|e^{i(1-s)X(t)}\, (A-B) \, e^{isX(t)}\|\, ds\: dt\\
        &\leq \int_{0}^{1}\|A-B\|dt\\
        &
        \leq \|A-B\|
    \end{split}
    \end{equation*}
    by~\ref{it: unitary invariance} and the Fundamental Theorem of Calculus.

    Item~\ref{it: norm on log} follows by noting that
    \begin{equation*}
    \begin{split}
    \|e^{iA}-I\|
    = \max_j |1-\exp(i\lambda_j)\|
    &= 2\max_j|\sin(\lambda_j/2)|,
    \end{split}
    \end{equation*}
    where $\lambda_j$ are the eingvalues of $A$. 
    Since $\max_j|\lambda_j|\leq \pi$ and by concavity of the sine function in the first quadrant, $\left|\sin x\right|\geq \frac{2|x|}{\pi}$ for all $x\in\left[-\frac{\pi}{2},\frac{\pi}{2}\right]$. Consequently,
    \begin{equation*}
    2 \frac{\|A\|}{\pi}\leq \|e^{iA}-I\|.
    \end{equation*}
\end{proof}

For later reference, we will find it useful to combine Lemma~\ref{lm: bch} with item~\ref{it: exponent approximation} as in
\begin{equation}\label{eq: bch}
    \left\|\prod_{j=1}^M \exp(i X_j)- \exp\left(\sum_{j=1}^MiX_j\right)\right\|\leq\frac{1}{2} \sum_{i<j} \|[X_i, X_j]\| + \|R\|. 
\end{equation}

\section{The Solovay-Kitaev Algorithm}\label{sec: sk section big}
In a nutshell, Solovay-Kitaev explores the local geometry of $\SU(d)$ to solve the gate synthesis problem and we will describe in this section.
We will be mostly following Dawson and Nielsen~\cite{dawson2005solovay}, with some extra care in getting exact bounds on the constants.
In fact, we conjecture that we are being unnecessarily conservative in our estimations of these prefactors and that the values reported by those authors are not only enough from a practical side, but theoretically sufficient as well.

\begin{algorithm}[H]
\caption{Solovay-Kitaev}
\label{alg: sk}
\KwIn{A target matrix $U\in \SU(d)$ and a depth parameter $t$}
\KwRequire{A universal instruction set $\{G_i\}$ and an $\varepsilon_0$-net, $\mathcal{N}_{\varepsilon_0}$}
\KwOut{A circuit approximation of $U$}
\If{$t==0$}{\Return \texttt{NET SEARCH}($U, \mathcal{N}_{\varepsilon_0}$)}
\Else{
$U_{t-1}\gets \texttt{SOLOVAY KITAEV}(U, t-1)$;

$V, W\gets \texttt{COMMUTATOR DECOMPOSITION}(\log(U^\dagger U_{n-1}))$

$V_{t-1}\gets \texttt{SOLOVAY KITAEV}(V, t-1)$;

$W_{t-1}\gets \texttt{SOLOVAY KITAEV}(W, t-1)$;

\Return $U_{t-1}W^\dagger_{t-1} V^\dagger_{t-1}W_{t-1} V_{t-1}$
}
\end{algorithm}

\begin{figure}
    \centering
    \includegraphics[width=0.8\textwidth]{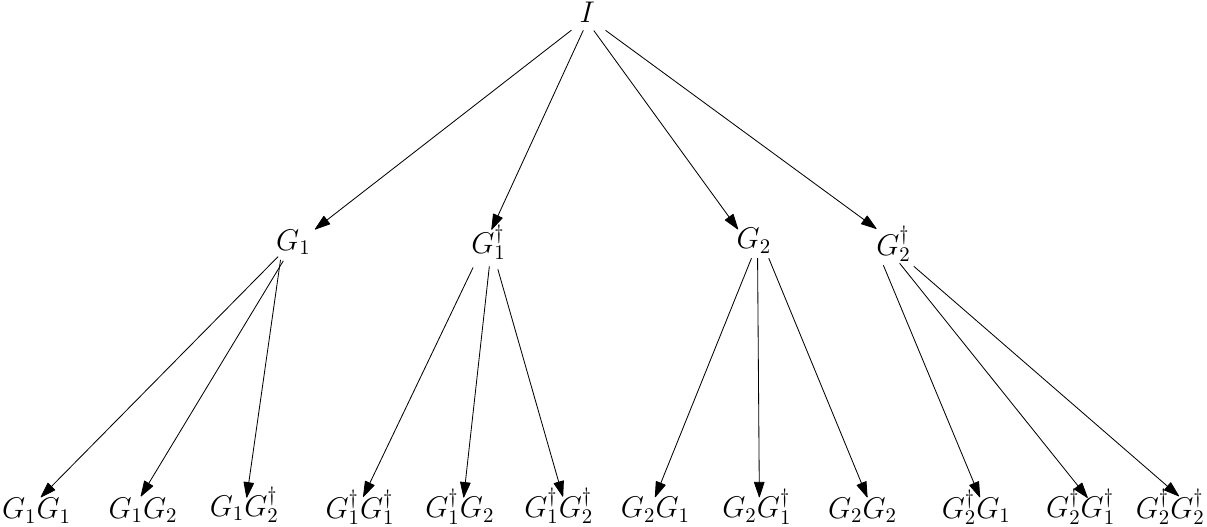}
     \caption{The first three layers of a tree structure used to build a net out of four instruction gates, $G_1, G_2, G_1^\dagger$, and $G_2^\dagger$. Trivial relations such as $G_1G_1^\dagger$, $G_2^\dagger G_2$, which equal the identity element $I$ (already in the uppermost layer of the tree) have been omitted.}
  \label{fig: tree}
\end{figure}

The algorithm is clearly recursive and, at depth zero, it executes a search over the $\varepsilon_0$-net, $\mathcal{N}_{\varepsilon_0}$.
In most of the descriptions, the net is assumed to be constructed through a brute force tree-like structure, where each layer is defined by appending a different instruction gate to the nodes of the previous one, see Figure~\ref{fig: tree}.
Some easy-to-implement operations help in decreasing the total number of nodes in this tree: for example, we avoid adding inverses that would lead to trivial cancellations, and, at each layer, we can apply some pruning operations to delete nodes that are too close to some other previous node of the net.
Of course, other relations could be taken into account to prune $\mathcal{N}_{\varepsilon_0}$ even further, but because the volume of a ball of radius $\varepsilon_0$ increases as $$\Theta(\varepsilon_0^{M})=\Theta(\varepsilon_0^{d^2-1})$$
the number of nodes in the net, $N_{\varepsilon_0}$, which, in the best case scenario, is {proportional to the manifold's volume divided by the balls' volume}, needs to grow \emph{at least} as $\Omega(\varepsilon_0^{-d^2+1})$.
Moreover, in total, the depth-zero case, that is, a search on the net, is called $3^t$ times during execution.
In practice, the number of calls is not as bad as it may look: assuming $0<\varepsilon<\varepsilon_0$ and $c^2\varepsilon_0<1$ where $c$ is a universal constant, the total depth of
\begin{equation}\label{eq: t}
    t = \left\lceil \frac{\log\left[\frac{\log(\varepsilon c^2)}{\log(\varepsilon_0 c^2)}\right]}{\log(3/2)}\right\rceil
\end{equation}
is enough to return a circuit approximation of error at most $\varepsilon$ from the target.
Still, because each net search may, in the worst case, inspect all $N_{\varepsilon_0}$ nodes, the whole algorithm is inefficient as function of the dimension.

When $t>0$, we proceed, on line 4, by recursively requesting an approximation of the target $U$ at depth $t-1$.
We compute the \emph{residual} $\Delta$ of this lower level approximation as the quantity
$\Delta_{t-1} =  U^\dagger\cdot U_{t-1}$,
where the product is taken through ordinary matrix multiplication.
By recursion, we may assume that $\|\Delta_{t-1}-I\|<\varepsilon_0$ and we will show that the algorithm converges exactly because, for all $t$, $$ \|\Delta_{t}-I\|<\|\Delta_{t-1}-I\|.$$
In other words, the \emph{residuals contract}.
We can then describe each step as finding new matrices, $V$ and $W$, whose group commutation is used to cancel out the residual of the previous step as much as possible.
That is, we want some $V$ and $W$ such that
\begin{equation}\label{eq: gc approximation}
    \|W^\dagger V^\dagger WV- \Delta_{t-1}^\dagger\|\leq \varepsilon,
\end{equation}
and
\begin{equation}\label{eq: pre triangle inequality}
    \|U_{t-1}\cdot W^\dagger V^\dagger WV-U\|= \|U(\Delta_{t-1}W^\dagger V^\dagger WV-I)\|=\|W^\dagger V^\dagger WV- \Delta^\dagger_{t-1}\|\leq \varepsilon,
\end{equation}
where we repeatedly used item~\ref{it: unitary invariance} of Lemma~\ref{lm: properties operator norm}.

\begin{algorithm}
\setcounter{AlgoLine}{0}
\caption{Commutator decomposition}
\label{alg: cd}
\KwIn{An input matrix $\Delta$.}
\KwOut{Two matrices $V$ and $W$ such that $\exp(i [\log V, \log W]) = \Delta$ and $\|\log V\|, \|\log W\|< 2\sqrt{\|\log \Delta\|}$.}

\textbf{function} $\texttt{COMMUTATOR DECOMPOSITION}(\Delta):$

$Z\gets \log \Delta$;

$Z, Q \gets \texttt{DIAGONALIZE}(Z)$ \tcp*{$Z$ is now diagonal and $P$ is the change of basis} 

$Z, Q \gets \texttt{BALANCE EIGENVALUES}(Z, Q)$ \tcp*{See Algorithm~\ref{alg: balance eigenvalue}}

\For{$k\gets 1$ \KwTo $d$}{$v_k\gets \sqrt{\frac{1}{2}\sum_{j=1}^kZ[j,j]}$;}

$X\gets \sum_{k=0}^{d-1} v_k \left( \ket{k}\bra{k+1} + \ket{k+1}\bra{k} \right)$;

$Y \gets i\sum_{k=0}^{d-1} v_k \left( \ket{k}\bra{k+1} - \ket{k+1}\bra{k} \right)$;

$X, Y \gets Q X Q^\dagger, Q Y Q^\dagger $ \tcp*{Convert back to the original basis}

\Return $\exp(i X), \exp(iY)$

\end{algorithm}

\begin{algorithm}
\setcounter{AlgoLine}{0}
\caption{Balance eigenvalues}
\label{alg: balance eigenvalue}
\KwIn{A diagonal matrix $Z=\diag(\mu_1,\dots,\mu_d)$ and a change of basis $Q$.}
\KwOut{A permutation $\tau$ of the diagonal entries of $Z$ with $0\leq \sum_{j=1}^kZ[j,j]\leq 2\|Z\|$ and $Q$ in this new basis.}

\textbf{function} $\texttt{BALANCE EIGENVALUES}(Z, Q):$

$P \gets [j : \mu_j \ge 0]$ \tcp*{Stack of indices of non-negative eigenavalues}
$N \gets [j : \mu_j < 0]$\tcp*{Stack of indices of negative eigenavalues}

$w \gets 0$\;
 $\tau \gets []$\;

\For{$k = 1$ to $d$}{

    \eIf{$w < \|Z\|$}{
        \eIf{$P \neq \emptyset$}{
            $j \gets \text{pop}(P)$\;
        }{
            $j \gets \text{pop}(N)$\;
        }
    }{
        \eIf{$N \neq \emptyset$}{
            $j \gets \text{pop}(N)$\;
        }{
            $j \gets \text{pop}(P)$\;
        }
    }

    append $j$ to $\tau$\;
    $w \gets w + \mu_j$\;
}

$Q\gets\texttt{PermutationMatrix}(Q, \tau)$ \tcp*{Transform $Q$ to the $\tau$-permuted basis}

\Return{$\diag(\mu_{\tau[1]},\dots, \mu_{\tau[d]})$, $Q$}

\end{algorithm}

For this, we linearize the problem to the Lie algebra, so we can use the \texttt{COMMUTATOR DECOMPOSITION} subroutine.
In particular, we want $V$ and $W$ to be such that
\begin{equation}\label{eq: commutator}
    [\log V, \log W] = \log (\Delta_{t-1}) \text{ and }\|\log V\|, \|\log W\|\leq 2 \sqrt{\|\log(\Delta_{t-1})\|}.
\end{equation}
By BCH, the first constraint on equation~\eqref{eq: commutator} implies that
\begin{equation}\label{eq: theoretical contraction}
    W^\dagger V^\dagger WV\approx \Delta_{t-1}^\dagger,
\end{equation}
whereas the second implies that the error of this approximation is small, namely $O(\|\log(\Delta_{t-1})\|^{3/2})$.
The function \texttt{COMMUTATOR DECOMPOSITION} and its \texttt{BALANCE EIGENVALUES} subroutine are described in Algorithms \ref{alg: cd} and \ref{alg: balance eigenvalue}, but basically they use simple linear algebra operations to find these two matrices $V$ and $W$.
In general, these steps are of time complexity $O(d^3)$ due to the diagonalization of the residual.

\begin{lemma}\label{lm: cd}  
Algorithm~\ref{alg: cd} converges in time $O(d^3)$, where the input $U$ is assumed with size $d\times d$.
\end{lemma}

Before proving this main lemma, however, we shall first demonstrate the following auxiliary result.

\begin{lemma}\label{lm: balance eigenvalues}
Suppose $Z$ is a $d\times d$ Hermitian matrix in diagonal form.
Then Algorithm~\ref{alg: balance eigenvalue} converges in time $O(d)$.
\end{lemma}
\begin{proof}
Because $Z$ is assumed in $\su(d)$,
\begin{equation}\label{eq: sum of eigenavalues}
\sum_j \mu_j = 0.
\end{equation}
The goal of the algorithm is to find a permutation of the eigenvalues of $Z$ for which, at each step (\ie, in each sublist of the permutation), the current sum of the eigenvalues does not exceed $2\|Z\|$.
Once we find such a permutation of eigenvalues, we apply the corresponding transformation to the columns of $Q$, basically reordering this diagonal matrix so that, again, the sum of the components up to each row is not larger than $2\|Z\|$.

Suppose that at the $k$-th step of the execution, $w_k = \sum_{j=1}^k \mu_{\tau(j)}$ where (with some abuse of notation) $\tau:\{1,\dots,d\}\to\{1,\dots,d\}$ is the permutation constructed so far, where we assume the invariant $0\leq w_{k-1}\leq 2\|Z\|$.
Equation~\eqref{eq: sum of eigenavalues} implies that the sum of the remaining eigenvalues not yet considered is equal to $-w_{k-1}$. 
Let $R$ be the set of indices of these remaining eigenvalues still to be considered.
We partition $R = R_{\geq 0}\cup R_{<0}$, where $R_{\geq 0}$ is the set of indices in $R$ representing the eigenvalues that are non-negative and $ R_{<0}$ is the set of indices representing the eigenvalues that are negative.

We now have to choose which value in $R$ we will assign to $w_k$.
We divide the decision cases. 
If $w_{k-1}< \|Z\|$, we look for a $\mu_j \geq 0$.
We know that $\mu_j\leq \|Z\|$, so $$w_k = w_{k-1}+\mu_j \leq 2\|Z\|.$$
If no such a $\mu_j$ is available, because $$\sum_{r\in R}\mu_r = -w_{k-1},$$ we have that $\mu_r \geq -w_{k-1}$.
Therefore, $$w_k = w_{k-1}+\mu_r\geq 0$$ for all $r\in R$.
Now, if $w_{k-1}\geq  \|Z\|$, either there exists a $-w_{k-1}<\mu_j<0$, or, if $ R_{< 0}=\emptyset$, then $ R_{\geq 0}$ is a list of zeros, which is impossible, as the remaining eigenvalues sum up to $-w_{k-1}<0$.

For the time complexity, we note that explicitly permuting the columns of $Q$ takes time $O(d^2)$, but if instead we simply redirect pointers of the columns, it can be accomplished in time $O(d)$.
\end{proof}

\begin{proof}[Proof of Lemma~\ref{lm: cd}]
    The correctness of the algorithm is shown in~\cite{kitaev2002classical}, but since we assume slightly different conventions, we will redo the demonstration here.
    For notation simplicity, we denote the input by $Z$ and by $\widetilde{Z}$ the same matrix in the basis of line 4; in particular, $Z=Q\widetilde{Z}Q^\dagger$. 
    
    One can easily verify the relation
    \begin{equation*}
        [X,Y]_{jk} = \begin{cases}
            v_j^2-v_{j-1}^2\text{ if } j=k\\
            0\text{ otherwise},
        \end{cases}
    \end{equation*}
    so that $[X,Y]=\widetilde{Z}$ where we use $v_k= \sqrt{\frac{1}{2}\sum_{j=1}^k\widetilde{Z_{jj}}}$.
    Because $$[QXQ^\dagger,QYQ^\dagger]=Q[X,Y]Q^\dagger =Q\widetilde{Z}Q^\dagger = Z$$
    the expected relation follows for the original basis as well.
    
    We now bound their norms.
    Since $X = A Y A^\dagger$, where $$A = \sum_{k=0}^{d-1} i^k \ket{k}\bra{k},$$ $\|X\|=\|Y\|$.
    By the Perron-Frobenius Theorem~\cite{meyer2023matrix}, $\|X\|\leq \|B\|$, where $$B= \sum_{k=0}^{d-2} \|Z\|^{1/2} \left( \ket{k}\bra{k+1} + \ket{k+1}\bra{k} \right).$$
    But $$\|B\|=2\|Z\|^{1/2}\cos\left(\frac{\pi}{d+1}\right)\leq 2\|Z\|^{1/2}$$ and the bound on $\|X\|$ follows.
    For the time complexity, it is enough to see that diagonalization is performed in $O(d^3)$.
\end{proof}

In theory, equation~\eqref{eq: theoretical contraction} already gives the sort of contraction on $\Delta_{t}$ that we need for the algorithm to converge, but we have a problem: although we do know the matrices $V$ and $W$ that we need, we do not yet have an implementation of them using instruction gates.
The trick is to recursively use SK to find these instructions as well, as shown in lines 6 and 7.
In particular, we will have, by recursion, that $V_{t-1}$ and $W_{t-1}$ approximate $V$ and $W$ with errors at most $\varepsilon_{t-1}$.
The surprising part is what follows from Lemma~\ref{lm: error contraction}: it is shown that
\begin{equation}\label{eq: multiplicative error compression}
    \|W^\dagger_{t-1} V^\dagger_{t-1}W_{t-1} V_{t-1}-W^\dagger V^\dagger W V\|< k\: \varepsilon_{t-1}^{3/2},
\end{equation}
where $k>0$ is a constant.
In other words, the whole algorithm works because we can use \emph{worse} approximations of the individual terms of the product $W^\dagger V^\dagger W V$ to \emph{improve} the approximation of the whole. 

 \begin{lemma}
 \label{lm: error contraction}
    Suppose $V_t, W_t$ are unitary approximations of $V$ and $W$ such that
    \begin{equation*}
        \|V-V_t\|,\|W-W_t\|\leq \varepsilon \text{ and } \|V-I\|,\|W-I\|\leq \delta,
    \end{equation*}
    then
    \begin{equation*}
        \|W^\dagger_{t} V^\dagger_{t}W_{t} V_{t}-W^\dagger V^\dagger W V\|\leq 8\varepsilon^2 + 8\varepsilon\delta + 4\varepsilon\delta^2 + 4\varepsilon^3 +\varepsilon^4.
    \end{equation*}
    
    In particular, let $\delta = k \sqrt{\varepsilon}$ where $k>0$ is a constant, $\varepsilon<1$ and
    \begin{equation*}
        \|V-I\|,\|W-I\|\leq k\sqrt{\varepsilon_{t-1}}.
    \end{equation*}
    Then
    \begin{equation*}
    \|W^\dagger_{t} V^\dagger_{t}W_{t} V_{t}-W^\dagger V^\dagger W V\|< (8k+4k^2+13)\varepsilon^{3/2}.
    \end{equation*}
 \end{lemma}

 \begin{proof}
    While the main idea of the proof comes from~\cite[Lemma 1]{dawson2005solovay}, we are more precise with bounds.
    Start by defining $\varepsilon_V, \varepsilon_W, \delta_V$, and $\delta_W$ as the matrices for which
    \begin{equation*}
    \begin{split}
       V_t = V + \varepsilon_V \;\;\; W_t = W + \varepsilon_W\\
       V = I + \delta_V \;\;\; W = I + \delta_W.
    \end{split}
    \end{equation*}
    We note that $\|\varepsilon_V\|,\|\varepsilon_W\|\leq \varepsilon$ and $\|\delta_V\|,\|\delta_W\|\leq \delta$.

    By expanding the commutator product, we see through direct computations that 
    \begin{equation*}
    \begin{split}
    \|W^\dagger_{t-1} V^\dagger_{t-1}W_{t-1} V_{t-1}-W^\dagger V^\dagger W V\| &= \|(W^\dagger+\varepsilon^\dagger_W)(V^\dagger+\varepsilon^\dagger_V)(W+\varepsilon_W)(V+\varepsilon_V)-W^\dagger V^\dagger W V\|\\
    &\leq \|\varepsilon_W^\dagger V^\dagger W V + W^\dagger V^\dagger \varepsilon_W V \| + \|W^\dagger \varepsilon_V^\dagger WV + W^\dagger V^\dagger W\varepsilon_V\| \\
    &\;\;+6\varepsilon^2+4\varepsilon^3+\varepsilon^4, 
    \end{split}
    \end{equation*}
    where we repeatedly applied the triangle inequality and item~\ref{it: unitary invariance} of Lemma~\ref{lm: properties operator norm}.
    We note, however, that
    \begin{equation*}
    \begin{split}
        \|\varepsilon_W^\dagger V^\dagger W V + W^\dagger V^\dagger \varepsilon_W V \|&=  \|\varepsilon_W^\dagger (I+\delta^\dagger_V) W (I+\delta_V) + W^\dagger (I+\delta^\dagger_V)\varepsilon_W (I+\delta_V) \|\\
        &\leq\|\varepsilon^\dagger_W W + W^\dagger \varepsilon_W\|+4\varepsilon\delta+2\varepsilon\delta^2,
    \end{split}
    \end{equation*}
    where we again used Lemma~\ref{lm: properties operator norm}.
    But, by unitarity of $W_t$
    \begin{equation*}
       I =  W_t^\dagger W_t=(W^\dagger+\varepsilon^\dagger_W)(W+\varepsilon_W) = I + \varepsilon^\dagger_W W + W^\dagger \varepsilon_W + \varepsilon^\dagger_W \varepsilon_W,
    \end{equation*}
    so
    \begin{equation*}
        \|\varepsilon_W^\dagger V^\dagger W V + W^\dagger V^\dagger \varepsilon_W V \| \leq \varepsilon^2 +4\varepsilon\delta + 2\varepsilon\delta^2.
    \end{equation*}
    Symmetrically,
    \begin{equation*}
        \|W^\dagger \varepsilon_V^\dagger WV + W^\dagger V^\dagger W\varepsilon_V\| \leq \varepsilon^2 +4\varepsilon\delta + 2\varepsilon\delta^2.
    \end{equation*}
    Therefore,
    \begin{equation*}
    \|W^\dagger_{t-1} V^\dagger_{t-1}W_{t-1} V_{t-1}-W^\dagger V^\dagger W V\| \leq 8\varepsilon^2 + 8\varepsilon\delta + 4\varepsilon\delta^2 + 4\varepsilon^3 +\varepsilon^4.
    \end{equation*}
    The second part follows directly by applying the hypothesis to the equation above.
 \end{proof}

We can now proceed to the SK theorem.
Here and throughout, we will denote by $T_0$, $S_0$, and $\ell_0$ the required time, space and expanded length of the circuit outputted by the base case subroutine of Algorithm~\ref{alg: sk} (lines 1 and 2).
Since these are defined by the net-search query, we recall that, as functions of $d$, both $S_0$ and $T_0$ scale as $O(d^2N_{\varepsilon_0})$, where $N_{\varepsilon_0}$ scales at least as $ 2^{\Omega(d^2)}$.

\begin{theorem}[Solovay-Kitaev]\label{th: sk full}
    Suppose $
     \varepsilon_0< \min\{1/c^2,\delta_0^2/32\pi\}$.
    Let $\{G_1,\dots,G_m\}$ be a universal instruction set for $\SU(d)$, and assume that words of length at most $\ell_0$ over this instruction set form an $\varepsilon_0$-net for $\SU(d)$.
    Take $U_1,\dots,U_p$ to be a list of unitaries in $\SU(d)$. 
    Then, using Algorithm~\ref{alg: sk},  one can construct circuits over $\{G_1,\dots,G_m\}$ approximating each $U_i$, $1\leq i \leq p$, up to error at most $\varepsilon$ in total time
    \begin{equation*}
        O((d^3+T_0)p\log^{k_t}(1/\varepsilon))
    \end{equation*}
 where the circuits' lengths are
 $$O(\ell_0\log^{k_\ell}(1/\varepsilon)),$$ with $k_t = \log(3)/\log(1.5)$, $k_\ell = \log(5)/\log(1.5)$.
  The working space, excluding the space needed to store the output straight-line program, is $O(S_0+\poly(d)\log\log(1/\varepsilon))$.
\end{theorem}

\begin{proof}
Assume $p=1$, the case for larger $p$ can be derived by repeated applications of this base case.
If $\varepsilon\geq \varepsilon_0$, use $t=0$ and we are done.
Otherwise, let $t$ be as in equation~\eqref{eq: t}.
Let $\Delta_t$ be the residual of $U_tU^\dagger$ at depth $t$, we will show that
\begin{equation*}
    \varepsilon_t = \|W_{t-1}^\dagger V_{t-1}^\dagger W_{t-1} V_{t-1} - \Delta_{t-1}^\dagger\|< c\: \varepsilon_{t-1}^{3/2},
\end{equation*}
which naturally implies convergence, since $\varepsilon_0 < 1/c^2$.

Start writing
\begin{equation*}
    \|W_{t-1}^\dagger V_{t-1}^\dagger W_{t-1} V_{t-1} - \Delta^\dagger_t\| \leq  \|W_{t-1}^\dagger V_{t-1}^\dagger W_{t-1} V_{t-1} -W^\dagger V^\dagger W V\|+\|W^\dagger V^\dagger W V-  \Delta^\dagger_t\|.
\end{equation*}
Note that, because $V$ and $W$ are outputs of Algorithm~\ref{alg: cd}, by the choice of $\varepsilon_0$,
\begin{equation*}
     \|V-I\| \leq \|\log V - \log I\| = \|\log V\|< 2 \sqrt{\|\log \Delta_{t-1}\|}\leq \sqrt{2\pi \varepsilon_{t-1}} 
\end{equation*}
by items~\ref{it: exponent approximation} and \ref{it: norm on log} of Lemma~\ref{lm: properties operator norm} and similarly for $W$.
Consequently, Lemma~\ref{lm: error contraction} applies with $\delta = \sqrt{2\pi \varepsilon_{t-1}}$ and $\varepsilon=\varepsilon_{t-1}$, and the first term of the right-hand side is bounded by $(8\sqrt{2\pi}+8\pi+13)\varepsilon_{t-1}^{3/2}$.

For the second term, we note that
$$
2\|\log V\|+2\|\log W\|
\leq
4\sqrt{2\pi\varepsilon_{t-1}}
\leq
4\sqrt{2\pi\varepsilon_0}
\leq
\delta_0,
$$
so Lemma~\ref{lm: bch} applies to the four factors defining the group
commutator. Therefore,
\begin{equation*}
\begin{split}
    \|W^\dagger V^\dagger W V - \Delta^\dagger_{t-1}\| & \leq (32\pi)^{3/2} C \varepsilon_{t-1}^{3/2}.
\end{split}
\end{equation*}
Letting $c = 8\sqrt{2\pi}+8\pi+13 + (32\pi)^{3/2}C$ and $\varepsilon_0 < 1/c^2$ enforces convergence.
In particular, for $t$ as in equation~\eqref{eq: t}, $\varepsilon_t\leq \varepsilon$. 

At each step, we note that, by recursion, $\ell_{t}\leq 5 \ell_{t-1}$, so 
$$\ell_t = O(\ell_0\log^{k_\ell}(1/\varepsilon)),$$
for $t$ as in equation~\eqref{eq: t}.
For the time complexity, we note that each step takes about three times more than the previous. 
Therefore, the total execution time is about $O(3^t (T_0+d^3))$, where the dependency in $d$ comes from Algorithm~\ref{alg: cd}, and $T_0$ is the time to execute \texttt{NET SEARCH}.
\end{proof}

\section{Lifting the curse of dimensionality in the Solovay-Kitaev's net queries}\label{sec: lifting}

\subsection{Net-free Solovay-Kitaev with trotterization}\label{sec: general algorithm}
Referring back to Algorithm~\ref{alg: sk}, we note that the \texttt{NET SEARCH} subroutine, whose objective is to find a circuit approximation of the current target within error $\varepsilon_0$, is the only place where $\mathcal{N}_{\varepsilon_0}$ is explicitly needed.
In this section, we describe a substitution of this particular step by a net-free approximation of the target.

To start, suppose that the Hermitian matrices $\{H_j\}_{j=1}^M$ form a basis for $\su(d)$, seen as a vector space---in particular, we can assume $M=d^2-1$.
We call a set of unitaries
\begin{equation}\label{eq: exponential basis}
    \calG_j = \exp(i H_j), \calG_j^\dagger = \exp(-i H_j) \quad\text{ for }1\leq j \leq M
\end{equation}
an \emph{exponential basis} for the Lie group $\SU(d)$.
For convenience, we will often omit the inverse $\calG_j^\dagger$ from exponential basis, although they should be always assumed to be present; in particular we denote solely by $\{\calG_j\}_{j=1}^M$ the basis of equation~\eqref{eq: exponential basis}. 
The rotation matrices
\begin{equation}\label{eq: pauli rot}
    R_X=e^{i\theta_X X}\quad R_Y=e^{i\theta_Y Y}\quad R_Z=e^{i\theta_Z Z},
\end{equation}
where $X, Y$ and $Z$ are the (Hermitian) Pauli matrices and $\theta_X, \theta_Y$, and $\theta_Z$ are some angles in $(0,\pi)$, form a perfectly fine example of an exponential basis for $\SU(2)$.

\emph{Product formulas}~\cite{childs2021theory} use exponential basis to approximate a target within arbitrary error $\varepsilon$ and have consequently received significant interest from the quantum computing community for their applications to the simulation of quantum systems~\cite{babbush2015chemical, childs2018toward, lloyd1996universal, wecker2015solving}.
Historically, the first product formula is what physicists know as \emph{trotterization}.
Explicitly, if $\{\calG_j\}_{j=1}^M=\{\exp(iH_j)\}_{j=1}^M$ forms an exponential basis, then there exist some real coefficients $\alpha_j$ such that
\begin{equation}\label{eq: basis decomp}
    \frac{H}{n} =  \sum_j \alpha_j H_j,
\end{equation}
where $n$ here is an integer chosen to ensure the BCH regime.
Defining the unitary
\begin{equation}\label{eq: trotterization}
    U^{(n)}=\left(\prod_{j=1}^M\exp\left(i\alpha_jH_j\right)\right)^n,
\end{equation}
and directly applying Lemma~\ref{lm: bch} implies that
\begin{equation*}
    \| U^{(n)}-U\|\leq  \frac{1}{n} \sum_{i<j}  \: |\alpha_i||\alpha_j| \:\|H_i\|\: \|H_j\| + O\left(\frac{1}{n^2}\right);
\end{equation*}
that is, by picking 
\begin{equation*}
n = \Omega\left(\frac{1}{\varepsilon_0}\sum_{i<j}  |\alpha_i||\alpha_j| \|H_i\|  \|H_j\|\right),
\end{equation*}
the sequence of unitaries in equation~\eqref{eq: trotterization} can be used as a substitute to \texttt{NET SEARCH} in Algorithm~\ref{alg: sk}.
Note, on the other hand, that the number of gates used in this approximation scheme scales linearly with $1/\varepsilon$, so, by itself, trotterization is unable to achieve the degree of efficiency, measured in circuit length, of the SK algorithm.

\begin{remark}
    Product formulas were later expanded beyond trotterization to include Suzuki formulas~\cite{suzuki1991general}, but we will only briefly discuss this alternative method in Section~\ref{sec: conclusion}.
\end{remark}

\begin{remark}
    Geometrically speaking, we note that product formulas serve as a \emph{numerical integration method} of the one-parameter subgroup $\gamma(t)=\exp(iHt)$, where $H$ is the logarithm of the target.
    The term ``integration" here is used to draw a parallel with numerical integration techniques to solve differential equations.
    Just as the Euler method approximates, with discrete steps, the smooth path drawn by the solution of an ordinary differential equation in configuration space, the idea here is to represent $\gamma$, which, in this case, is a geodesic of a bi-invariant metric, as a sequence of applications of instruction gates.
    In this sense, the choice to use product formulas is only but a convenience: any other geodesic of a metric in $\SU(d)$ that can be efficiently approximated by a product of gates would work.

    Our approach should be then contrasted with the one taken by Nielsen and collaborators~\cite{nielsen2006geometric, nielsen2006optimal, nielsen2006quantum}.
    There, instead of an easy-to-integrate metric such as bi-invariant ones, they search for some Riemannian metric (or, more generally, for a right-invariant cost function) that can be used to bound $\ell$ in some meaningful way.
    In this case, the metric works as a lower bound on the number of instruction gates needed to implement a unitary within a given accuracy.
    This methodology, is better aligned with the field of \emph{optimal geometric control} or, using more modern terminology, of \emph{optimal quantum control}~\cite{d2021introduction, dong2010quantum}, but, unfortunately, their constructions are usually non-explicit, meaning that although they do find a bound on $\ell$, it is usually not trivial to get the explicit sequence of instruction gates directly.
    It would be interesting to find an example of cost that bounds $\ell$, but is, at the same time, numerically integrable.
\end{remark}

The problem with trotterization, however, is that equation~\eqref{eq: basis decomp} necessarily assumes the coefficients $\alpha_j$ to be real numbers, where, in practice, even if we suppose to have access to physical realizations of $\{\calG_j\}_{j=1}^M$ as gates, only integer powers
$$
\exp(i c_j H_j)=\calG_j^{c_j},
\qquad c_j\in\mathbb Z,
$$
are directly implementable in a quantum circuit (negative powers are assumed to be implemented using the inverse gates).
That is, we need a \emph{discrete version} of trotterization.
This approach is summarized in Algorithm~\ref{alg: depth zero}.

\begin{algorithm}[H]
\setcounter{AlgoLine}{0}
\caption{Discrete trotterization}
\label{alg: depth zero}
\KwIn{A target matrix $U=\exp(iH)$.}
\KwRequire{An exponential basis of $\SU(d)$, $\{\mathcal{G}_j\}_{j=1}^M$, with $H_j=\log \mathcal{G}_j$ and a positive integer $n$.}

\KwOut{A sequence of powers of $\mathcal{G}_j$ that approximates $U$.}

\textbf{function} $\texttt{NET FREE SEARCH}(U, n):$

$P \gets [H_1 \; H_2 \; \dots \; H_M]$\;

$\alpha_j \gets \texttt{BASIS CHANGE}(P, H/n)$ \tcp*{Write $H/n$ in the basis $\{H_j\}$} 

\For{$j \gets 1$ \KwTo $M$}{
$c_j\gets \texttt{round}(\alpha_j)$
}

\Return $U^{(n)} = \left[\prod_{j=1}^M \mathcal{G}_j^{c_j}\right]^n$

\end{algorithm}

The algorithm takes as input a target unitary matrix $U=\exp(iH)$, assumed to be given in the usual basis for the $d\times d$ matrices, $M_d(\C)$.
Line 3 of Algorithm~\ref{alg: depth zero} uses simple linear algebra to compute the real coefficients $\alpha_j$ by applying the precomputed change-of-basis matrix from the fixed coordinates of $\mathfrak{su}(d)$ inherited from $M_d(\C)$ to the basis $\{H_j\}_{j=1}^M$.
Because instead of the $\alpha_j$, we use the closest integers $c_j$, we note that
\begin{equation}\label{eq: c minus alpha}
    \|\alpha_j H_j-c_j H_j\|\leq \frac{1}{2} \|H_j\|,
\end{equation}
so, if we can make $\|H_j\|$ small enough, we can control the error introduced by approximating the real coefficients with integers only.
Unfortunately, naively setting 
$h_{\max}=\max_j \|H_j\|\leq {\varepsilon_0}/{M}$
is usually not enough: at each of the $n$ steps of the form $\prod_{j=1}^M \calG_j^{c_j}$, the error due to the discretization of coefficients can accumulate, creating what we call a \emph{drift}.
This means that, so the drift is guaranteed not to grow too much, we actually need to enforce 
$h_{\max}\leq {\varepsilon_0}/{(nM)}$,
as we will see in Proposition~\ref{prop: depth zero}.
Since $n$ will be large (equation~\eqref{eq: n}), $h_{\max}$ will tend to be too small.
We note, on the other hand, that these bounds on $h_{\max}$ might be avoidable if we allow the coefficients $c_j$ to vary at each step, so that the drift portion of the error never grows too fast.
We return to this possibility in the conclusion.

Careful analysis shows that the change-of-basis preprocessing costs $O(M^3)=O(d^6)$ time and 
$O(M^2)=O(d^4)$ space.
After this preprocessing, each call to the base routine costs $O(d^3+M^2)=O(d^4)$ time in straight-line output representation.
If, instead, we require an uncompressed sequence, the execution time is the often larger $O\left(n\sum_j|c_j|\right)$.
Here and throughout, we will take $M=d^2-1$ and recall that $\varepsilon_0\leq \delta_0$.

\begin{proposition}\label{prop: depth zero}
    Suppose $\{\calG_j\}_{j=1}^M$ is a fixed exponential basis for $\SU(d)$.
    Define
    \begin{equation}\label{eq: def mu}
         \mu = \inf_{{\alpha}\neq 0}\frac{\|\sum_{j=1}^M \alpha_j H_j\|}{\;\;\;\|{\alpha}\|_\infty}
    \end{equation}
    where $\alpha=(\alpha_1,\dots,\alpha_M)$ is a vector of coefficients in $\R^M$, 
    \begin{equation}\label{eq: def dimensionless gain}
        \mu_N = \frac{\mu}{h_{\max}},
    \end{equation}
    and $h_{\max} = \max_j \|H_j\|$.
    For some fixed value of $\varepsilon_0$, let $n$ be a positive integer such that 
    \begin{equation}\label{eq: n}
        n\geq\frac{\pi^2M^2K}{\varepsilon_0\mu_N^2},
    \end{equation}
    where $K$ is some universal constant.
    Call $U^{(n)}$ the output of Algorithm \ref{alg: depth zero} for some input $U\in\SU(d)$ using the basis $\{\calG_j\}_{j=1}^M$.
    If
    \begin{equation}\label{eq: bound h_max prop}
        h_{\max}\leq \frac{\varepsilon_0}{nM},
    \end{equation}
    then $\|U^{(n)} -U\|\leq\varepsilon_0$.
\end{proposition}

Before we proceed to the proof, some comments are due.
The variable $\mu$ is called the \emph{minimum gain} of the basis $\{H_j\}_{j=1}^M$ and can be seen as a generalization of the usual term in numerical linear algebra: it basically measures how large the coefficients of a linear combination of the basis $\{H_j\}_{j=1}^M$ need to be to be able to reach any \emph{direction} in the vector space.
For example, if one used a Hilbert space norm on $\mathfrak{su}(d)$ and measured coefficients in $\ell_2$, the analogous minimum gain would be the smallest singular value of $A$, or the square root of the smallest eigenvalue of the corresponding Gram matrix~\cite{horn2012matrix}.
In particular, the more all basis vectors in $\{H_j\}_{j=1}^M$ point in the same direction, the smaller $\mu$ will be.
This notion can be made a little bit more rigorous thanks to the simple lemma below.

\begin{lemma}\label{lm: mu larger than zero}
A set of $M=d^2-1$ Hermitian traceless matrices $\{H_j\}_{j=1}^M$ forms a basis of $\su(d)$ if and only if $\mu>0$.
\end{lemma}

\begin{proof}
If the $H_j$'s do not form a basis, then they are linearly dependent. Hence there exists $\alpha\neq 0$ such that
$$
\sum_{j=1}^M\alpha_jH_j=0.
$$
Therefore,
$$
0\leq \mu
\leq
\frac{
\left\|\sum_{j=1}^M\alpha_jH_j\right\|
}{
\|\alpha\|_\infty
}
=0,
$$
so $\mu=0$.

Conversely, suppose that $\{H_j\}_{j=1}^M$ forms a basis.
Then the linear map
$$
    A:\mathbb R^M\to\mathfrak{su}(d),
    \qquad
    A\alpha=\sum_{j=1}^M\alpha_jH_j,
$$
is injective. 
The function $$\alpha\mapsto \|A\alpha\|$$ is continuous, and the set
$$
    \{\alpha\in\R^M|\|\alpha\|_\infty=1\}
$$
is compact.
Since $A\alpha\neq 0$ on this set, the minimum of $\|A\alpha\|$ over it is strictly positive. By homogeneity, this minimum is exactly $\mu$.
\end{proof}

The normalized quantity $\mu_N=\mu/h_{\max}$ removes the overall scale of the basis and measures only its conditioning.
Note that $0< \mu_N\leq 1$ and that $\min_{j}\|H_j\|=h_{\max}$ is a necessary but not sufficient conditions for this maximum to be achieved.

\begin{proof}[Proof of Proposition~\ref{prop: depth zero}]
For the sake of notation simplicity, let
$$
S=\prod_{j=1}^M e^{ic_jH_j},
\qquad
T=e^{iH/n}
=
\exp\left(i\sum_{j=1}^M\alpha_jH_j\right).
$$
In particular, the output of Algorithm~\ref{alg: depth zero} is $U^{(n)}=S^n$, while $U=T^n$.
By telescoping and unitary invariance of the operator norm,
\begin{equation*}
    \begin{split}
\|U^{(n)}-U\|
&=
\|S^n-T^n\| \\
&\leq
\sum_{r=0}^{n-1}
\|S^{n-r-1}(S-T)T^r\| \\
&\leq
n\|S-T\|.
\end{split}
\end{equation*}

We decompose the one-step error as
\begin{equation*}
    \begin{split}
\|S-T\|
&\leq
\left\|
\prod_{j=1}^M e^{ic_jH_j}
-
\exp\left(i\sum_{j=1}^M c_jH_j\right)
\right\| 
+
\left\|
\exp\left(i\sum_{j=1}^M c_jH_j\right)
-
\exp\left(i\sum_{j=1}^M \alpha_jH_j\right)
\right\|.
\end{split}
\end{equation*}

The second term corresponds to the drift error and is bounded by
\begin{equation*}
\begin{split}
\left\|
\exp\left(i\sum_{j=1}^M c_jH_j\right)
-
\exp\left(i\sum_{j=1}^M \alpha_jH_j\right)
\right\|
&\leq
\left\|
\sum_{j=1}^M(c_j-\alpha_j)H_j
\right\| \\
&\leq
\sum_{j=1}^M |c_j-\alpha_j|\|H_j\| \\
&\leq
\frac{Mh_{\max}}2\\
&
\leq
\frac{\varepsilon_0}{2n},
\end{split}
\end{equation*}
where the last inequality uses $h_{\max}\leq \varepsilon_0/(nM)$.

It remains to bound the product-formula error. By the definition of $\mu$, since
$$
\sum_{j=1}^M\alpha_jH_j=\frac{H}{n},
$$
we have
$$\|\alpha\|_\infty
\leq
\frac{\|H\|}{n\mu}.$$
Moreover, because $c_j$ is obtained, for each $1\leq j\leq M$, by nearest-integer rounding, assuming $|\alpha_j|\geq 1/2$,
$$|c_j|\leq \frac{\|H\| }{n \mu}+\frac{1}{2}<\frac{2\|H\| }{n \mu}.$$
This is a deliberately wasteful step: it ignores cancellations and the fact that many small coefficients round to zero. The resulting bound should therefore be understood as a conservative sufficient condition for Algorithm~\ref{alg: depth zero} to be an \(\varepsilon_0\)-base routine, not as an optimized analysis of the discretization error.

Still, we have that
$$
\delta
=
\sum_{j=1}^M\|c_jH_j\|
\leq Mh_{\max} \times \frac{2\|H\|}{n\mu}\leq
\frac{2M\|H\|}{n\mu_N}
\leq
\frac{4M\pi}{n\mu_N}.$$
By the assumed lower bound on $n$, this quantity is at most $\delta_0$, so BCH applies.
Using that, for any $A$ and $B$ matrices, $\|[A,B]\|\leq 2\|A\|\|B\|$,
we obtain
\begin{equation*}
\begin{split}
\left\|
\prod_{j=1}^M e^{ic_jH_j}
-
\exp\left(i\sum_{j=1}^M c_jH_j\right)
\right\|
&\leq
\frac{1}{2}
\sum_{i<j}\|[c_iH_i,c_jH_j\|
+
C\delta^3 \\
&\leq
\sum_{i<j}\|c_iH_i\|\|c_jH_j\|
+
C\delta^3\\
&\leq
\frac{1}{2}
\left(\sum_{j=1}^M\|c_jH_j\|\right)^2
+
C\delta^3\\
&\leq
\frac{\delta^2}{2}
+
C\delta^3.
\end{split}
\end{equation*}

Since $
\delta\leq {4M\pi}/{n\mu_N}$,
the choice of $n$ implies, after increasing the universal constant in~\eqref{eq: n} if necessary, that
\begin{equation}\label{eq: bound K}
    \frac{\delta^2}{2}+C\delta^3
\leq
\frac{\varepsilon_0}{2n}.
\end{equation}
Combining this estimate with the rounding-error bound gives the result.
\end{proof}

\begin{remark}\label{rm: bound on K}
    Equation~\eqref{eq: bound K} implies that $K \geq 32\sqrt{C}$ is enough.
\end{remark}

The choice of the Trotter parameter $n$ is constrained in two opposite directions.
Making $n$ lower bounded as in equation~\eqref{eq: n} for a given basis $\{\calG_j\}_{j=1}^M$ is easy: simply let $n$ as in equation~\eqref{eq: n}.
Nonetheless, because of bounds in $h_{\max}$, we also need to ensure that
    $n\leq {\varepsilon_0}/{h_{\max} M},$
that is, just asking for a large $n$ is not enough: if, for example, 
\begin{equation*}
    \frac{\varepsilon_0}{h_{\max} M} \leq \frac{\pi^2M^2 K}{\varepsilon_0\mu_N^2},
\end{equation*}
Algorithm~\ref{alg: depth zero} could yield an approximation of $U$ within  a distance larger than $\varepsilon_0$.

A sufficient way to avoid this obstruction is to impose a smallness condition on $h_{\max}$.
In particular, to ensure that 
\begin{equation*}
    n < \frac{\pi^2M^2K}{\varepsilon_0\mu_N^2}+ 1 \leq \frac{\varepsilon_0}{h_{\max} M},
\end{equation*}
where we assume that $\mu_N$ depends only on $d$, we can take
\begin{equation}\label{eq: good basis equation}
    h_{\max} \leq \frac{\varepsilon_0^2\mu_N^2}{\pi^2M^3K+\varepsilon_0\mu_N^2M}.
\end{equation}
This inspires the next definition.

\begin{definition}[Good exponential basis]
    An exponential basis $\{\calG_j\}_{j=1}^M=\{\exp(iH_j)\}_{j=1}^M$ of $\SU(d)$ is called \emph{good} if $h_{\max}$ is bounded as in equation~\eqref{eq: good basis equation}.
\end{definition}

We see that, at least at this level, assuming $\mu_N\sim \Theta(1)$, $h_{\max}\lesssim M^{-3}\sim d^{-6}$, which is very impractical. 
We will discuss in the conclusion strategies for picking slightly larger values of $h_{\max}$ while still guaranteeing convergence, but for now, we note that we can finally state and prove our modification to the SK algorithm, provided we have a good exponential basis.

\begin{algorithm}[H]
\setcounter{AlgoLine}{0}
\caption{Modified Solovay-Kitaev}
\label{alg: modified sk}
\LinesNumbered
\KwIn{A target matrix $U\in \SU(d)$, with $H=\log U$, and a depth parameter $t$.}
\KwRequire{A good exponential basis $\{\calG_j\}_{j=1}^M=\{\exp(iH_j)\}_{j=1}^M$ of $\SU(d)$.}
\KwOut{A circuit approximation of $U$.}

\textbf{function} $\texttt{SOLOVAY KITAEV}(U, t):$

\eIf{$t==0$}{
{
$n\gets 
\left\lceil \frac{\pi^2M^2K}{\varepsilon_0\mu_N^2} \right\rceil$

\Return $\texttt{NET FREE SEARCH}(U, n)$}
}{
$U_{t-1}\gets \texttt{SOLOVAY KITAEV}(U, t-1)$;

$V, W\gets \texttt{COMMUTATOR DECOMPOSITION}(\log(U^\dagger U_{t-1}))$;

$V_{t-1}\gets \texttt{SOLOVAY KITAEV}(V, t-1)$;

$W_{t-1}\gets \texttt{SOLOVAY KITAEV}(W, t-1)$;

\Return $U_{t-1}W^\dagger_{t-1} V^\dagger_{t-1}W_{t-1} V_{t-1}$
}
\end{algorithm}

\begin{theorem}[Modified Solovay-Kitaev]\label{th: modified sk}
    Suppose that $\{\calG_j\}_{j=1}^M$ is a good exponential basis of $\SU(d)$.
    Take $U_1,\dots,U_p$ to be a list of unitaries in $\SU(d)$. 
    Then Algorithm~\ref{alg: modified sk} converges to a circuit approximation of the targets $U_i$, $1\leq i\leq p$, with error at most $\varepsilon$.
    In particular, after an $O(d^6)$-time preprocessing step for the change of basis, the running time is
    $O\left(pd^4\log^{k_t}(1/\varepsilon)\right),$
    so the total time including preprocessing is
    $$
    O(d^6+pd^4\log^{k_t}(1/\varepsilon)).
    $$
    The additional working space, excluding the output circuit, is $O(\poly(d)\log\log(1/\varepsilon))$.
    The circuits' lengths satisfy
    $$
    O\left(\frac{d^2}{\mu}\log^{k_\ell}(1/\varepsilon)\right),
    $$
    where
    $k_t={\log 3}/{\log(3/2)}$ and
    $k_\ell={\log 5}/{\log(3/2)}$.
\end{theorem}
\begin{proof}
    Algorithm~\ref{alg: modified sk} is exactly equal to Algorithm~\ref{alg: sk}, except for the subroutine described between lines 3 and 4, therefore, we must only show correctness and complexity of this part.
    Nevertheless, because the exponential basis is assumed to be good, by Proposition~\ref{prop: depth zero}, just like \texttt{NET SEARCH}, these will return an approximation of $U$ with the gates $\mathcal{G}_j$ of error at most $\varepsilon_0$, so correctness will follow.
    For the complexity analysis, the change-of-basis matrix used by Algorithm~\ref{alg: depth zero} is computed once, in time $O(M^3)=O(d^6)$. After this preprocessing, each depth-zero call takes $O(d^3+M^2)=O(d^4)$
    time.
    Since every recursive call generates three subcalls, the recursion tree has $O(3^t)$ nodes.
    Therefore, for $p$ targets, the total running time after preprocessing is
    $$
    O\left(pd^4 3^t\right)
    =
    O\left(pd^4\log^{k_t}(1/\varepsilon)\right),
    $$
    where $k_t=\log(3)/\log(3/2)$. Including the preprocessing step, the total running time is thus
    $$
    O\left(d^6+pd^4\log^{k_t}(1/\varepsilon)\right).
    $$
    Finally, the length of the sequence is given by $5^t\ell_0$, where $\ell_0$ is now the length of a sequence constructed using Algorithm~\ref{alg: depth zero}, namely, of size $n\sum_{j=1}^M|c_j|=O\left(\frac{M}{\mu}\right)$.
\end{proof}

Comparing the asymptotic behavior of $d$ in Theorems~\ref{th: sk full} with that of Theorem \ref{th: modified sk}, we see that the exponential gain in the complexities of time and space is balanced by a potential loss in the final size of the circuits.
The ``potential" here is crucial: Theorem~\ref{th: sk full} uses a \emph{lower-bound} on $\ell_0$ based on the already-mentioned measure-theoretical pressure towards exponentially more gates to form an $\varepsilon_0$-net for $\SU(d)$, but the quality of the instruction sets might force much worse scaling in practice.
There are a few metrics in the literature to measure the quality of nets~\cite{bourgain2012spectral, dankert2009exact, oszmaniec2021epsilon} and, under suitable assumptions on their values, it is possible to get an upper bound scaling of $\ell_0$ on $\poly(d)$ as well~\cite{harrow2002efficient}.
On the other hand, in Theorem~\ref{th: modified sk}, we have an upper bound, but this also assumes a \emph{fixed} $\mu_N$.
That is, $\mu_N$ can be seen as a very simple hyperparameter to measure the quality of a basis set, and, just as in the instruction sets' case, letting it run free could lead to blow-ups. 
The use of condition numbers for characterizing the quality of generating sets is well-known in control theory, but as far as we can tell, its application to the SK setting is new.

\begin{remark}
    Note that Theorem~\ref{th: modified sk} provides a verifiable sufficient condition for checking if an instruction set is universal.
    In fact, in the real-RAM model, given a set of gates $\calG_k$, this condition can be verified by computing the logarithms $H_j$, checking that they form a basis of $\su(d)$, estimating their minimum gain $\mu$, and verifying the required bound on $h_{\max}$.
    On the other hand, this criterion is only sufficient: many universal instruction sets need not form good exponential bases, for example because their logarithms are too large or do not themselves constitute a basis of $\su(d)$. 
\end{remark}

\subsection{Generating a good exponential basis from a net}\label{sec: basis from net}
Theorem~\ref{th: modified sk} says that if we have a good exponential basis, we can skip the whole \texttt{NET SEARCH} subroutine of SK and define an algorithm that is also efficient as a function of $d$.
However, although in some practical cases we might be given a good basis as an instruction set (for example, hardware implementations of the Pauli rotations of equation~\eqref{eq: pauli rot} with small angles are sometimes available~\cite{moflic2026constant, smith2025optimally, vandaele2024optimal}), we cannot always take a general universal set $\{G_1,\dots, G_m\}$ to be good.  
On the other hand, we can always \emph{approximate} a good exponential basis $\{\calG_j\}_{j=1}^M$ using a universal set as a preprocessing step based on the vanilla SK algorithm. 
In this section, we will explore the details of this subroutine.
For such, we will need to define two new concepts.

First, we need to consider the \emph{Frobenius inner product} as an easy-to-compute alternative to the operator norm.
If $A$ and $B$ are two $d\times d$ matrices, their Frobenius inner product is
\begin{equation*}
    \langle A|B \rangle_{\text{F}} = \text{tr}(A^\dagger B).
\end{equation*}
Restricted to the real vector space of Hermitian traceless matrices, this is a real inner product.
It induces the Frobenius norm
\begin{equation*}
    \|A\|_F=\sqrt{\operatorname{tr}(A^\dagger A)},
\end{equation*}
which i equivalent to the operator norm and satisfies
$$
\|A\|\leq \|A\|_F\leq \sqrt{d}\|A\|
$$
for every $d\times d$ matrix $A$.

Second, we consider the \emph{generalized Gell-Mann matrices} in $\su(d)$~\cite{bertlmann2008bloch, georgi2000lie, kimura2003bloch}, a Frobenius orthonormal basis of $\su(d)$.
These matrices are divided into three sectors, the so-called \emph{symmetric sector}
\begin{equation*}
   \Lambda_{k\ell}^{(s)}
=
\frac{1}{\sqrt 2}
\bigl(
\ket{\ell}\bra{k}
+
\ket{k}\bra{\ell}
\bigr)\text{ for } 1\leq k<\ell\leq d,
\end{equation*}
the \emph{antisymmetric} sector
\begin{equation*}
\Lambda_{k\ell}^{(a)}
=
\frac{i}{\sqrt 2}
\bigl(
\ket{\ell}\bra{k}
-
\ket{k}\bra{\ell}
\bigr) \text{ for } 1\leq k<\ell\leq d,
\end{equation*}
and the \emph{Cartan sector},
\begin{equation*}
    \Lambda_\ell^{(c)}
=
\sqrt{\frac{1}{\ell(\ell+1)}}
\left(
\sum_{k=1}^\ell \ket{k}\bra{k}
-
\ell\ket{\ell+1}\bra{\ell+1}
\right)
\text{ for } 1\leq \ell\leq d-1.
\end{equation*}
Besides forming a basis for $\su(d)$, they have large enough $\mu$, as stated in the lemma below.

\begin{lemma}\label{lm: GM}
   Let $\{\Lambda_j\}_{j=1}^M$ denote the generalized Gell-Mann basis defined above.
   Fix a scale $h^*>0$, and set
   $H_j^*=h^*\Lambda_j$.
   Then, as a basis for $\su(d)$, $\{H_j^*\}_{j=1}^M$ has minimum gain
    $$
    \mu^*\geq \frac{h^*}{\sqrt{d}}.
    $$ 
    Consequently, if $h_{\max}^*=\max_j\|H_j^*\|$, then the dimensionless minimum gain satisfies
    $$
    \mu_N^*
    =
    \frac{\mu^*}{h_{\max}^*}
    \geq
    \frac{1}{\sqrt d}.
    $$
\end{lemma}
\begin{proof}
    For the first part,
    \begin{equation*}
        \mu^* = \inf_{\|\alpha\|_{\infty}=1} \left\|\sum_j \alpha_j H^*_j\right\| \geq \inf_{\|\alpha\|_{\infty} =1}\frac{1}{\sqrt{d}} \left\|\sum_j \alpha_j H^*_j\right\|_{\text{F}} = \inf_{\|\alpha\|_{\infty} =1}\frac{1}{\sqrt{d}} \sqrt{\sum_j |\alpha_j|^2 \|H^*_j\|_F^2}\geq \frac{h^*}{\sqrt{d}}
    \end{equation*}
    where, in the last inequality, we used that $\|\alpha\|_2\geq \|\alpha\|_\infty$ for any vector $\alpha\in\R^M$.

    For the second part, notice that $h_{\max}^*=\max_j\|h^*\Lambda_j\|\leq h^*$, because each Gell-Mann matrix has operator norm at most $1$. Therefore
    $$
    \mu_N^*
    =
    \frac{\mu^*}{h_{\max}^*}
    \geq
    \frac{h^*}{h^*\sqrt{d}}
    =
    \frac{1}{\sqrt d}.
    $$
\end{proof}

Given a universal set $\{G_i\}_{i=1}^m$, our strategy will consist of finding circuits $\{\calG_j\}_{j=1}^M$ that approximate $\{\calG^*_j\}_{j=1}^M=\{\exp(iH^*_j)\}_{j=1}^M$.
Since $\{\calG^*_j\}_{j=1}^M$ is a good exponential basis when we take, for example,
\begin{equation}\label{eq: bound h star}
h_{\max}^*\leq \frac{\varepsilon_0^2}{16d\pi^2M^3K}\leq \frac{\varepsilon_0^2\mu_N^2}{16\pi^2M^3K}\sim O(d^{-7}),    
\end{equation}
$\{\calG_j\}_{j=1}^M$ will also be good if we assume a fine-enough approximation error (Theorem~\ref{th: GM}).
The details are carried out in the demonstration of the next two results.

\begin{theorem}\label{th: GM}
    Let, for $1\leq j \leq M$, $H^*_j = h^*_{\max}\Lambda_j$ for some scalar for $h^*_{\max}$ as in equation~\eqref{eq: bound h star}.
    Suppose that $\{\mathcal{G}_j\}_{j=1}^M$ are unitary approximations of $\{\mathcal{G}_j^*\}_{j=1}^M = \{\exp(i H_j^*)\}_{j=1}^M$ in $\SU(d)$.
    If
    \begin{equation}\label{eq: delta}
     \delta= \max_j\|\mathcal{G}^*_j- \mathcal{G}_j\|\leq \frac{h_{\max}^*}{2\pi d^{5/2}}\sim O(d^{-19/2}),
    \end{equation}
    then 
    \begin{equation}\label{eq: estimation h_max}
        h_{\max}=\max_j\|H_j\|\leq h^*_{\max}\left(1+\frac{1}{2d^{5/2}}\right)\leq \frac{\varepsilon^2_0}{8\pi^2dM^3K},
    \end{equation}
    where $\{H_j\}_{j=1}^M=\{\log \mathcal{G}_j\}_{j=1}^M$, and its dimensionless minimum gain, $\mu_N={\mu}/{h_{\max}}$, is such that $\mu_N  \geq {1}/(2\sqrt{d})$.
    In particular, $\{\mathcal{G}_j\}_{j=1}^M$ is a good exponential basis for $\SU(d)$.
\end{theorem}
\begin{proof}
  We start showing that $h_{\max}\leq \pi$, so item~\ref{it: norm on log} of Lemma~\ref{lm: properties operator norm} applies.
  In particular, letting $\widetilde{H}_j$ be the principal branch logarithms of $\calG_j$ with eigenvalues $\widetilde{\theta_i}$ as in the proof of Lemma~\ref{lm: new bound}, we note that,
  \begin{equation*}
      2\left|\sin\left(\frac{\widetilde{\theta_i}}{2}\right)\right| = |e^{i\widetilde{\theta_i}}-1|\leq \|\calG_j-1\|\leq \|\mathcal{G}_j-\mathcal{G}^*_j\|+\|\mathcal{G}^*_j-I\| \leq \delta+h^*_{\max} < \frac{4}{d}\leq 2 \sin\left(\frac{\pi}{d}\right),
  \end{equation*}
  where we again used concavity of sine. 
  Because $|\widetilde{\theta_j}|\leq \pi$, this gives
  \begin{equation}\label{eq: bound on 2pi d}
      |\widetilde{\theta_i}|< \frac{2\pi}{d}
  \end{equation}
  where $d\geq 2$.
  This means that
  \begin{equation*}
      \left|\sum_{i=1}^d\widetilde{\theta_i}\right|\leq \sum_{i=1}^d|\widetilde{\theta_i}|< d \frac{2\pi}{d}\leq 2\pi,
  \end{equation*}
  which, together with equation~\eqref{eq: mod eigenvalues}, implies that $\widetilde{H}_j$ is traceless for all $1\leq j\leq M$, and $\widetilde{H}_j=H_j$.
  
  Consequently, by item~\ref{it: norm on log} and \ref{it: exponent approximation}
    \begin{equation}\label{eq: bad approximation}
        h_{\max} \leq \frac{\pi}{2}\max_j \{\|\mathcal{G}_j-I\|\}\leq  \frac{\pi}{2} \max_j \{\|\mathcal{G}_j-\mathcal{G}^*_j\|+\|\mathcal{G}^*_j-I\| \}\leq \frac{\pi}{2}(\delta+h^*_{\max})\leq \pi h_{\max}^*,
    \end{equation}
    but, this level of approximation is too weak for our purposes.
    Instead, let $E_j=H_j-H_j^*$, where $\delta'=\max_j\|E_j\|$, and so, by equation~\eqref{eq: bad approximation},
    \begin{equation}\label{eq: bad approx delta prime}
      \delta' = \max_j\|H_j-H_j^*\| \leq
    h_{\max}+h_{\max}^* \leq (\pi+1)h_{\max}^* \leq 5h_{\max}^*.
    \end{equation}
    
    Now set $\Omega_j = \log(\mathcal G_j(\mathcal G_j^*)^\dagger)$; note that 
    \begin{equation*}
        \|\exp(i\Omega_j)-I\|\leq \|G_j(\mathcal G_j^*)^\dagger-I\|\leq \|(G_j-\mathcal G_j^*)(\mathcal G_j^*)^\dagger\mathcal\|=\|G_j-\mathcal G_j^*\|\leq \delta 
    \end{equation*}
    and the same argument of equation~\eqref{eq: bound on 2pi d} applies.
    In particular, $\|\Omega_j\| < \pi$ and, again, by item~\ref{it: norm on log},
    \begin{equation}\label{eq: bound on omega}
      \|\Omega_j\| \leq \frac{\pi}{2} \|\mathcal G_j(\mathcal G_j^*)^\dagger-I\| = \frac{\pi}{2} \|\mathcal G_j-\mathcal G_j^*\| \leq \frac{\pi}{2}\delta.  
    \end{equation}
    On the other hand,
    $$
    \calG_j(\calG_j^*)^\dagger =  e^{iH_j}e^{-iH_j^*}  = e^{i(H_j^*+E_j)}e^{-iH_j^*}.
    $$
    Since
    $$\|H_j\|+\|H_j^*\|\leq \pi h^*_{\max} + h^*_{\max} \leq \varepsilon_0\leq\delta_0,$$ 
    BCH applies; and
    \begin{equation*}
    \begin{split}
    \|\Omega_j\| &= \left\|E_j + \frac{1}{2}[(H_j^*+E_j),-H_j^*] + R_j\right\|\\
    &= \left\|E_j -\frac{1}{2}[E_j,H_j^*] + R_j\right\|\\
    &\geq \|E_j\|-h_{\max}^*\|E_j\|-C(h_{\max}^*+\|E_j\|)^2\|E_j\|,
    \end{split}
    \end{equation*}
    by Remark~\ref{rm: bch}. 
    Using once more the crude bound of equation~\eqref{eq: bad approx delta prime}, this gives
    $$
    \|\Omega_j\|\geq  \|E_j\|-\|E_j\| h_{\max}^*- 36(h_{\max}^*)^2 \|E_j\| = \left(1-h_{\max}^*-36C(h_{\max}^*)^2\right)\|E_j\|.
    $$
    By choosing $h_{\max}^*$ sufficiently small, the term $h_{\max}^*+36C(h_{\max}^*)^2$ is at most $1/2$. Hence,
    \begin{equation*}
       \delta'=\max_j\|E_j\| \leq 2\|\Omega_j\|\leq \pi \delta\leq \frac{h_{\max}^*}{2d^{5/2}},
    \end{equation*}
    by equation~\eqref{eq: bound on omega}.
    So,
    \begin{equation*}
        h_{\max}\leq h^*_{\max} + \delta' =h^*_{\max}\left(1+\frac{1}{2d^{5/2}}\right).
    \end{equation*}

    Moving to the minimum gain, we note that
    \begin{equation*}
    \mu
    =
    \inf_{\|\alpha\|_\infty=1}
    \left\|\sum_j\alpha_jH_j\right\|
    \geq
    \inf_{\|\alpha\|_\infty=1}
    \left\|\sum_j\alpha_jH_j^*\right\|
    -
    M\delta'
    \geq
    \mu^*-M\delta'.
    \end{equation*}
    The definition of dimensionless gain therefore gives
    \begin{equation*}
    \begin{split}
        \mu_N &= \frac{\mu}{h_{\max}}\\
        &\geq \frac{\mu^*-M\delta'}{h^*_{\max} + \delta'}\\
        & = \frac{\mu^*}{h^*_{\max} + \delta'} - \frac{M\delta'}{h^*_{\max} + \delta'}\\
        &\geq \frac{\mu^*}{h^*_{\max}}-\frac{\delta'}{h^*_{\max}}\left(\frac{\mu^*}{h^*_{\max}}+M\right)\\
        & \geq \frac{1}{\sqrt{d}} -\frac{\delta'}{h^*_{\max}}\left(1+M\right)\\
        &\geq \frac{1}{\sqrt{d}} -\frac{d^2}{2d^{5/2}}\\
        &\geq \frac{1}{2\sqrt{d}},
    \end{split}
    \end{equation*}
    where, in the fourth line, we used $\frac{1}{1+x}\geq 1-x$ for all $x\in \R_+$ and $\delta'\geq 0$, and Lemma~\ref{lm: GM} and $\mu^*\leq h^*_{\max}$ in the fifth.
    Since $\mu_N>0$, by Lemma~\ref{lm: mu larger than zero}, $\{H_j\}_{j=1}^M$ is a basis for $\su(d)$.
    Moreover, that $\{\calG_j\}_{j=1}^M$ forms a good exponential basis follows from equation~\eqref{eq: good basis equation} with the estimated values of $h_{\max}$ and $\mu_N$.
\end{proof}

\begin{corollary}\label{col: sk}
Suppose that $\{G_1,\dots,G_m\}$ is a universal instruction set for $\SU(d)$, and that words of length at most $\ell_0$ over this instruction set form an $\varepsilon_0$-net for $\SU(d)$.
Let $T_0$ and $S_0$ denote, respectively, the time and space costs of the corresponding depth-zero net-search routine as in Theorem~\ref{th: sk full}.

Take $U_1,\dots,U_p$ to be a list of unitaries in $\SU(d)$. 
Then, after a preprocessing step based on the usual Solovay-Kitaev algorithm, one can construct a circuit over $\{G_1,\dots,G_m\}$ approximating each $U_i$, $1\leq i\leq p$, up to error at most $\varepsilon$ in total time
$$
O( d^2(T_0+d^3)\log^{k_t} d + p d^4\log^{k_t} (1/\varepsilon)),
$$
using $O(S_0+d^4)$ additional working space excluding the storage needed for the output circuits. 
The expanded circuit lengths are
$$
O\left(\ell_0 d^{19/2} \log^{k_\ell}(d)
\log^{k_\ell}(1/\varepsilon) \right),$$
where $k_t={\log 3}/{\log(3/2)}$ and $k_\ell={\log 5}/{\log(3/2)}$.
\end{corollary}
\begin{proof}
    By the standard Solovay-Kitaev (Theorem~\ref{th: sk full}), approximations $\{\calG_j\}_{j=1}^M$ of the unitaries $\{\calG^*_j\}_{j=1}^M$ of Theorem~\ref{th: GM} in terms of the gates $\{G_1,\dots, G_m\}$ with error $\delta$ bounded as equation~\eqref{eq: delta} can be computed in time
    \begin{equation*}
        O((T_0+d^3)\log^{k_t}(1/\delta))=O\left((T_0+d^3)\log^{k_t}(d^{19/2})\right)=O((T_0+d^3)\log^{k_t}(d)),
    \end{equation*}
    for each $1\leq j \leq M$.
    Moreover, each such approximation has, at most,
    $$
    O(\ell_0\log^{k_\ell}(d))
    $$
    many gates.
    But, by Theorem~\ref{th: GM}, $\{\calG_j\}_{j=1}^M$ is a good exponential basis of $\SU(d)$, so Theorem~\ref{th: modified sk} applies.
    The total time to approximate the input targets is
    \begin{equation*}
        O(d^6+p d^4\log^{k_t}(1/\varepsilon)),
    \end{equation*}
    where the $d^6$ preprocessing step of trotterization might as well be absorbed into the $T_0$ SK cost.
    Assuming, 
    $$
    h_{\max}^* = \frac{\varepsilon_0^2}{16d\pi^2M^3K} = \Theta(d^{-7}),
    $$
    the total lengths of the output circuits, in terms of the original universal gates, are
    \begin{equation*}
    \begin{split}
         O\left(\frac{d^2\ell_0\log^{k_\ell}(d)}{\mu}\log^{k_\ell}(1/\varepsilon)\right)&=O\left(\frac{d^2\ell_0\log^{k_\ell}(d)}{\mu_Nh_{\max}}\log^{k_\ell}(1/\varepsilon)\right) \\
         &= O(\ell_0 d^{19/2}\log^{k_\ell}(d)\log^{k_\ell}(1/\varepsilon)),
    \end{split}
    \end{equation*}
    where we used the estimation of $\mu_N$ and $h_{\max}\geq h^*_{\max}-\delta'\sim h^*_{\max}$ in Theorem~\ref{th: GM}.
\end{proof}

\section{Conclusion}\label{sec: conclusion}
Table~\ref{tb: algorithms} compares the three versions of SK discussed in this paper.
In terms of execution time and dedicated space only, the good basis solution of Algorithm~\ref{alg: modified sk} and the general method of Corollary~\ref{col: sk} are preferable to Algorithm~\ref{alg: sk}, especially when $p$ is large.
Yet, it should be noted that in the case of Corollary~\ref{col: sk}, although moved to a $p$-independent preprocessing step, the linear scaling with $T_0$ and, consequently, the exponential dependence on $d^2$, is still present.
Hence, a natural extension of our approach consists of searching for an alternative to the construction of a good exponential basis $\{\calG_j\}_{j=1}^M$ from the instruction set $\{G_1,\dots, G_m\}$ that does not depend on the original SK algorithm or any other explicit enumeration of the net's nodes.
We do not know of any existing algorithm that achieves this and, again quoting \cite{kitaev2002classical}, this seems to be a rather complicated problem. 
Still, we conjecture that, if $\ell_\delta$ is the minimum circuit length such that the set of words on $\{G_1,\dots, G_m\}$ forms a $\mathcal{N}_{\delta}$-net in $\SU(d)$, there exists an algorithm that constructs a good exponential basis in time polynomial in the net depth $\ell_{\delta}$ and in $d$, without enumerating the exponentially large set of words of length at most $\ell_{\delta}$.
In particular, for $d$-indexed families of instruction gates such that $\ell_\delta=O(\poly(d, 1/\delta))$, Corollary~\ref{col: sk} can be modified to imply on a $O(\poly(d, \log(1/\varepsilon)))$ solution to gate synthesis.

\begin{table}
\renewcommand{\arraystretch}{1.5}
\begin{tabular}{c c c c}
\hline
Method& Constraints  & Time complexity & $\ell$ \\  
& on instructions & & \\
\hline
Algorithm~\ref{alg: sk} & None & ${O}((d^3+T_0)p\log^{k_t}(1/\varepsilon))$ & $O(\ell_0\log^{k_\ell}(1/\varepsilon))$\\
 Algorithm~\ref{alg: modified sk} & Good exponential basis & $  O(d^6+pd^4\log^{k_t}(1/\varepsilon))$ & $O\left(\frac{d^2}{\mu}\log^{k_\ell}(1/\varepsilon)\right)$\\
 Corollary~\ref{col: sk} & None & $\widetilde{O}( d^2(T_0+d^3) + p d^4\log^{k_t} (1/\varepsilon))$ & $\widetilde{O}\left(\ell_0 d^{19/2}
\log^{k_\ell}(1/\varepsilon) \right)$\\
 \hline
\end{tabular}
\caption{Comparison between the three distinct modifications of the SK algorithm discussed in this paper. Here, $d$ is the dimension of the qudits, $\varepsilon$ the target error level, $k_t={\log 3}/{\log(3/2)}$, $k_\ell={\log 5}/{\log(3/2)}$, $p$ the number of unitaries being approximated, $\mu$ the minimum gain of the logarithms of a good exponential basis, and $T_0$ the time to query an $\varepsilon_0$-net. The symbol $\widetilde{O}$ is used to hide polylogarithmic factors on $d$.}\label{tb: algorithms}
\end{table}

Still, even if this net free solution exists, the parameter $\ell$ is now severely penalized (as a function of the dimension) in comparison with the usual SK.
Although we used some very coarse approximations to estimate it, it should be noted that $\ell$'s bad dependence on the dimension comes from the term 
\begin{equation*}
    n\sum_{j=1}^M|c_j|\leq M n\max_j |c_j|=O\left(\frac{d^2}{\mu}\right)=O\left(\frac{d^2}{h_{\max}\mu_N}\right)
\end{equation*}
which measures the number of gates due to the trotterization step.
In the present construction, the dimensionless gain $\mu_N^{-1}$ contributes only with a factor of $O(d^{1/2})$, whereas the dominant cost comes from $h^{-1}_{\max}$, which scales as $\Theta(d^{7})$.

The dependence of $h_{\max}$ on the dimension can, in principle, be amortized through the Suzuki formulas.
These alternative integration methods for the one-parameter subgroup use higher order terms to asymptotically decrease the integration error, with (sometimes) only an acceptable cost in the total circuit length. 
For example, the Suzuki formula at level $k=1$, also known as \emph{Strang's splitting}~\cite{strang1968construction}, uses second order cancellations to improve the integration error rate and, consequently, the value of $n$ by a factor of square root, with an extra price paid in doubling the number of applied gates.
The formulas are recursively defined in $k$ and improve the product-formula error at the cost of increasing the number of exponentials per step, typically exponentially in the order. 
Still, however, the bounds on the formulas' remainders usually follow from BCH, so one needs to ensure that its radius of convergence is respected, that is
\begin{equation*}
    \delta = \sum_{j=1}^M \|c_jH_j\|\leq M\max_j|c_j|h_{\max}\leq \delta_0. 
\end{equation*}
This suggests a product-formula barrier around $h_{\max}\sim M^{-2}\sim d^{-4}$, which would still leave a length contribution of order $d^{6}$ up to conditioning factors.

Moreover, in our case, besides the usual integration errors, there is also drift.
One could, for example, introduce a \emph{scheduled} integration scheme.
Formally, this means that the coefficients $c_i$ are now allowed to adjust at each integration step, so that the drift error never increases beyond some constant bound.
Using the pushforward of the right and left actions, one may describe the problem of finding an optimal schedule as an ordinary closest vector problem in the lattice at each step.
Although an optimal schedule could be, in general, \texttt{NP}-hard~\cite{micciancio2002complexity, vanemdeboas1981another}, efficiently computing sufficiently good approximated solutions, in which the drift error never increases beyond a certain constant limit, may be possible.
We leave this for future work.

Finally, our method uses a bi-invariant metric because its geodesics are explicit and compatible with product formulas.
Other geometries, including Nielsen-type right-invariant or even Finsler metrics, may better capture gate count but are harder to integrate algorithmically.
Designing efficiently integrable metrics whose lengths correlate with discrete gate complexity is an interesting open direction.

\bibliographystyle{plain}
\bibliography{net}

\end{document}